\documentclass[12pt,draftcls,onecolumn]{IEEEtran}

\usepackage{subfigure}
\usepackage{amssymb, amsmath, amsfonts}
\usepackage{empheq}
\usepackage{caption, cite}

\usepackage{algorithm}
\usepackage{algorithmic}

\usepackage{amsthm}
\usepackage{tikz}

\usetikzlibrary{external}
\usepackage{pgfplots}
\usepackage{cases}

\newtheorem{theorem}{Theorem}
\newtheorem{definition}{Definition}
\newtheorem{lemma}{Lemma}

\newtheorem{corollary}{Corollary}
\newtheorem{remark}{Remark}
\newtheorem{assumption}{Assumption}

\newlength\figureheight
\newlength\figurewidth

\ifnum\pdfshellescape=1
\fi

\DeclareFontFamily{OT1}{pzc}{}
\DeclareFontShape{OT1}{pzc}{m}{it}{<-> s * [1.000] pzcmi7t}{}
\DeclareMathAlphabet{\mathpzc}{OT1}{pzc}{m}{it}

\newcommand{\A}{{\mathcal{A}}}

\newcommand{\R}{{\mathbb{R}}}
\newcommand{\E}{{\mathpzc{E}}}
\newcommand{\bS}{{\mathpzc{S}}}

\newcommand{\I}{{\mathcal{I}}}
\newcommand{\M}{{\mathcal{M}}}

\DeclareMathOperator{\inte}{int}

\DeclareMathOperator{\bal}{Bal}
\DeclareMathOperator{\ebal}{EBal}

\DeclareMathOperator{\dom}{dom}
\DeclareMathOperator{\supp}{supp}

\newcommand\addtag{\refstepcounter{equation}\tag{\theequation}}

\makeatletter
\newcommand{\rmnum}[1]{\romannumeral #1}
\newcommand{\Rmnum}[1]{\expandafter\@slowromancap\romannumeral #1@}
\makeatother

\title{{Binary Hypothesis Testing with Byzantine Sensors:
Fundamental Trade-off Between Security and Efficiency}}

\author{Xiaoqiang Ren$^*$, Jiaqi Yan$^*$, and Yilin Mo$^{*\dag}$
\thanks{$*$: Xiaoqiang Ren, Jiaqi Yan and Yilin Mo are with the School of Electrical and Electronics Engineering, Nanyang Technological University, Singapore. Emails: {xren@ntu.edu.sg,\,jyan004@e.ntu.edu.sg,\,ylmo@ntu.edu.sg}}%
  \thanks{$\dag$: Corresponding Author.}
}

\begin{document}
\maketitle
\begin{abstract}
  This paper studies binary hypothesis testing based on measurements from a set of sensors, a subset of which can be compromised by an attacker.
  The measurements from a compromised sensor can be manipulated arbitrarily by the adversary.
  The asymptotic exponential rate, with which the probability of error goes to zero, is adopted to indicate the detection performance of a detector.
  In practice, we expect the attack on sensors to be sporadic, and therefore the system may operate with all the sensors being benign for an extended period of time.
  This motivates us to consider the trade-off between the detection performance of a detector, i.e., the probability of error,  when the attacker is absent (defined as efficiency) and the worst-case detection performance when the attacker is present (defined as security).
  We first provide the fundamental limits of this trade-off, and then propose a detection strategy that achieves these limits.
  We then consider a special case, where there is no trade-off between security and efficiency. In other words, our detection strategy can achieve the maximal efficiency and the maximal security simultaneously.
  Two extensions of the secure hypothesis testing problem are also studied and fundamental limits and achievability results are provided: 1) a subset of sensors, namely ``secure'' sensors, are assumed to be equipped with better security countermeasures and hence are guaranteed to be benign; 2) detection performance with unknown number of compromised sensors.
  Numerical examples are given to illustrate the main results.

\end{abstract}
\begin{IEEEkeywords}
  Hypothesis testing, Security, Secure detection, Efficiency, Trade-off, Byzantine attacks, Fundamental limits
\end{IEEEkeywords}

\section{Introduction}
\emph{Background and Motivations:}
Network embedded sensors, which are pervasively used to monitor the system, are vulnerable to malicious attacks due to their limited capacity and sparsely spatial deployment.
An attacker may get access to the sensors and send arbitrary messages, or break the communication channels between the sensors and the system operator to tamper with the transmitted data.
Such integrity attacks have motivated many researches on how to infer useful information from corrupted sensory data in a secure manner~\cite{rawat2011collaborative, mo2015secure,teixeira2015secure}.
In this paper, we follow this direction but with the focus on the trade-off between the performance of the inference algorithm when the attacker is absent and the ``worst-case'' performance when the attacker, which has the knowledge of the inference algorithm, is present.
We define two metrics, \emph{efficiency} and \emph{security}, to characterize the performance of the hypothesis testing algorithm (or detector) under the two scenarios respectively and analyze the trade-off between security and efficiency.

\emph{Our Work and its Contributions:}
We consider the sequential binary hypothesis testing based on the measurements from $m$ sensors.
It is assumed that $n$ out of these $m$ sensors may be compromised by an attacker, the set of which is chosen by the attacker and fixed over time. The adversary can manipulate the  measurements sent by the compromised sensors arbitrarily. According to Kerckhoffs's principle~\cite{shannon1949communication}, i.e., the security of a system should not rely on its obscurity, we assume that the adversary knows exactly the hypothesis testing algorithm used by the fusion center.
On the other hand, the fusion center (i.e., the system operator) only knows  the number of malicious sensors $n$, but does not know the exact set of the compromised sensors.

At each time $k$, the fusion center needs to make a decision about the underlying hypothesis based on the possibly corrupted measurements collected from all sensors until time $k$.
Given a hypothesis testing algorithm at the fusion center (i.e., a measurements fusion rule), the worst-case probability of error is investigated, and the asymptotic exponential decay rate of the error, which we denote as the ``security'' of the system, is adopted to indicate the detection performance.
{On the other hand, when the attacker is absent, the detection performance of a hypothesis testing algorithm, i.e., the asymptotic exponential decay rate of the error probability,  is denoted by the ``efficiency".  }

We focus on the trade-off between efficiency and security.
In particular, we are interested in characterizing the fundamental limits of the trade-offs between efficiency and security and the detectors that achieve these limits. 

The main contributions of this work are summarized as follows:
\begin{enumerate}
  \item To the best of our knowledge, this is the first work that studies the trade-off between the efficiency and security of any inference algorithm.
  \item With mild assumptions on the probability distributions of the measurements, we provide the fundamental limits of the trade-off between the efficiency and security (Corollaries~\ref{corollary:S-E}~and~\ref{corollary:E-S}).       
      Furthermore, we present detectors, with low computational complexity, that achieve these limits (Theorem~\ref{theorem:achieve}).
    Therefore, the system operator can easily adopt the detectors we proposed to obtain the best trade-off between efficiency and security.
    Interestingly, in some cases, e.g., Gaussian random variables with same variance and different mean, the maximal efficiency and the maximal security can be achieved simultaneously (Theorem~\ref{theorem:SEsimultaneously}).
  \item Similar results, i.e., the fundamental limits of the trade-off and the detectors that possess these limits, are established with several different problem settings (Section~\ref{sec:extension}).
    This shows that our analysis techniques are insightful and may be helpful for the future related studies.
\end{enumerate}

\emph{Related Literature:}
{
A sensor is referred to as a Byzantine sensor if its messages to the fusion center are fully controlled by an adversary\footnote{{In practice, to manipulate the data of a sensor, an adversary may attack the sensor node itself or break the communication channel between the sensor and the fusion center. In this paper, we do not distinguish these two approaches.} }.
Recently,
detection with Byzantine sensors has been studied extensively in~\cite{marano2009distributed,kailkhura2015distributed, kailkhura2015asymptotic, abdelhakim2014distributed, soltanmohammadi2013decentralized,   soltanmohammadi2014fast, mo2014resilient,  vamvoudakis2014detection, abrardo2016game, Jiaqi_CDC_detection}, among which~\cite{marano2009distributed, kailkhura2015distributed, kailkhura2015asymptotic} took the perspective of an attacker and aimed to find the most effective attack strategy,~\cite{abdelhakim2014distributed, soltanmohammadi2013decentralized, soltanmohammadi2014fast, mo2014resilient} focused on designs of resilient detectors, and~\cite{vamvoudakis2014detection, abrardo2016game, Jiaqi_CDC_detection} formulated the problem in a game-theoretic way.
The main results of~\cite{marano2009distributed} are the critical fraction of Byzantine sensors that blinds the fusion center, which is just the counterpart of the breakdown point in robust statistics~\cite{huber2011robust}, and the most effective attack strategy that minimizes the asymptotic error exponent in the Neyman-Pearson setting, i.e., the Kullback--Leibler (K--L) divergence. Since the Byzantine sensors are assumed to generate independent and identical distributed (i.i.d.) data, the resulting measurements with minimum K--L divergence and the corresponding robust detector coincide with those in~\cite{huber1965robust}. Similar results were obtained in~\cite{kailkhura2015distributed, kailkhura2015asymptotic} by considering non-asymptotic probability of error in Bayesian setting and asymptotic Bayesian performance metric, i.e., Chernoff information, respectively. The authors in~\cite{abdelhakim2014distributed} focused on computation efficient algorithms to determine optimal parameters of the $q$-out-of-$m$ procedure~\cite{viswanathan1989counting} in large scale networks for different fractions of Byzantine sensors. More than two types of sensors were assumed in~\cite{soltanmohammadi2013decentralized, soltanmohammadi2014fast}. The authors thereof proposed a maximum likelihood  procedure, which is based on the iterative expectation maximization (EM) algorithm~\cite{dempster1977maximum}, simultaneously classifying the sensor nodes and performing the hypothesis testing. The authors in~\cite{mo2014resilient} showed that the optimal detector is of a threshold structure when the fraction of Byzantine sensors is less than $0.5$. A zero-sum game was formulated in each of~\cite{vamvoudakis2014detection, abrardo2016game, Jiaqi_CDC_detection}, among which a closed-form equilibrium point of attack strategy and detector was obtained in~\cite{Jiaqi_CDC_detection}, computation efficient and nearly optimal equilibrium point (exact equilibrium point only in certain cases) was obtained in~\cite{vamvoudakis2014detection}, and numerical simulations were used to study the equilibrium point in~\cite{abrardo2016game}.

While in~\cite{marano2009distributed,kailkhura2015distributed, kailkhura2015asymptotic, abdelhakim2014distributed, soltanmohammadi2013decentralized,   soltanmohammadi2014fast, abrardo2016game} the Byzantine sensors are assumed to generate malicious data independently, this work, as in~\cite{mo2014resilient,  vamvoudakis2014detection, Jiaqi_CDC_detection}, assumes that the Byzantine sensors may collude with each other. The collusion model is more reasonable since the attacker is malicious and will arbitrarily change the messages of the sensors it controlls. Notice also that compared to the independence model, the collusion model complicates the analysis significantly. Unlike~\cite{kailkhura2015distributed, kailkhura2015asymptotic, abdelhakim2014distributed, soltanmohammadi2013decentralized,   soltanmohammadi2014fast, abrardo2016game, vamvoudakis2014detection}, where the sensors only send binary messages, this work, as in~\cite{marano2009distributed,mo2014resilient, Jiaqi_CDC_detection}, assumes that the measurements of a benign sensor can take any value.
Since the binary message model simplifies the structure of corrupted measurements, and, hence, implicitly limits the capability of an attacker, it is easier to be dealt with. This work differs from~\cite{mo2014resilient, Jiaqi_CDC_detection} as follows. The authors in~\cite{mo2014resilient} focused on one time step scenario. The analysis is thus fundamentally different and more challenging. On the contrary, in this work the hypothesis testing is performed sequentially and an asymptotic regime performance metric, i.e., the Chernoff information, is concerned.
A similar setting as in this work was considered in our recent work~\cite{Jiaqi_CDC_detection}. However,~\cite{Jiaqi_CDC_detection} focused on the equilibrium point. The performance (i.e., the security and efficiency) of the obtained equilibrium detection rule is merely one point of the admissible set that will be characterized in this paper.

Finally, we should remark that the aforementioned literature mainly focuses on designing algorithms in adversarial environment. However, those algorithms may perform poorly in the absence of the adversary comparing to the classic Neyman-Pearson detector or Naive Bayes detector. A fundamental question, which we seek to answer in this paper, is that how to design a detection strategy which performs ``optimally'' regardless of whether the attacker is present.


}

\emph{Organization:} In Section~\ref{sec:ProblemFormulation}, we formulate the problem of binary hypothesis testing in adversarial environments, in which the attack model, the performance indices and the notion of the efficiency and security are defined.
For the sake of completeness, we give a brief introduction the large deviation theory in Section~\ref{sec:prilimilary}, which is a key supporting technique for the later analysis.
The main results are presented in Section~\ref{sec:fundamentalLimitation}.
We first provide the fundamental limits of the trade-off between the efficiency and security.
We then propose detectors that achieve these limits.
At last, we show that the maximal efficiency and the maximal security can be achieved simultaneously in some cases. 
Two extensions are investigated in Section~\ref{sec:extension}.
After providing numerical examples in Section~\ref{sec:simulation}, we conclude the paper in Section~\ref{sec:conclusion}.

\emph{Notations}: $\R$ ($\mathbb{R}_{+}$) is the set of (nonnegative) real numbers.
$\mathbb{Z}_{+}$ is the set of positive integers.
The cardinality of a finite set $\mathcal I$ is denoted as $|\mathcal I|$.
For a set $\mathcal{A}\in\mathbb R^n$, $\inte(\mathcal{A})$ denotes its interior.
For any sequence $\{x(k)\}_{k=1}^\infty$, we denote its average at time $k$ as $\bar x(k) \triangleq \sum_{t=1}^k x(t)/k$. {For a vector $\mathbf x\in\R^n$, the support of $\mathbf x$, denoted by $\supp(\mathbf x)$, is the set of indices of nonzero elements:
\[\supp(\mathbf x) \triangleq \{i\in\{1,2,\ldots,n\}: \mathbf x_i \neq 0\}.\]}

\section{Problem Formulation} \label{sec:ProblemFormulation}
\label{sec:problem}
Consider the problem of  detecting a binary state $\theta \in \{ 0,\,1\}$ using $m$ sensors' measurements.
Define the measurement $\mathbf{y}(k)$ at time $k$ to be a row vector:
\begin{equation}
  \mathbf{y}(k)\triangleq \begin{bmatrix}
    y_1(k)&y_2(k)&\cdots&y_m(k)
  \end{bmatrix}\in \mathbb R^m,
  \label{eq:yvector}
\end{equation}
where $y_i(k)$ is the scalar measurement from sensor $i$ at time $k$.
For simplicity, we define $\mathbf{Y}(k)$ as a vector of all measurements from time $1$ to time $k$:
\begin{equation}
  \mathbf{Y}(k) \triangleq \begin{bmatrix}
    \mathbf{y}(1)&\mathbf{y}(2)&\cdots&\mathbf{y}(k)
  \end{bmatrix}\in \mathbb R^{mk}.
  \label{eq:bigY}
\end{equation}
Given $\theta$, we assume that all measurements $\{y_i(k)\}_{i=1,\ldots,m,\,k=1,2,\ldots}$ are independent and identically distributed (i.i.d.).
The probability measure generated by $y_i(k)$ is denoted as $\nu$ when $\theta = 0$ and it is denoted as $\mu$ when $\theta = 1$.
In other words, for any Borel-measurable set $\A\subseteq \mathbb R$, the probability that $y_i(k)\in \A$ equals $\nu(\A)$ when $\theta = 0$ and equals $\mu(\A)$ when $\theta = 1$.
We denote the probability space generated by all measurements $\mathbf{y}(1),\,\mathbf{y}(2),\,\dots$ as $(\Omega_y,\,\mathcal F_y,\,\mathbb P^{o}_\theta)$\footnote{The superscript ``o" stands for original, which is contrasted with corrupted measurements.}, where for any $l\geq 1$
\begin{align*}
  \mathbb P^{o}_\theta(y_{i_1}(k_1)\in \A_1&,\dots,  y_{i_l}(k_l)\in \A_l) \\
  &= \begin{cases}
    \nu(\A_1)\nu(\A_2)\dots\nu(\A_l)&\text{if }\theta = 0\\
    \mu(\A_1)\mu(\A_2)\dots\mu(\A_l)&\text{if }\theta = 1
  \end{cases},
\end{align*}
when $(i_j,k_j)\neq (i_{j'},k_{j'})$ for all $ j \neq j' $.
The expectation taken with respect to $\mathbb P^{o}_\theta$ is denoted by $\mathbb E^{o}_\theta$.
We further assume that $\nu$ and $\mu$ are absolutely continuous with respect to each other.
Hence, the log-likelihood ratio $\lambda:\mathbb R \rightarrow \mathbb R$ of $y_i(k)$ is well defined as
\begin{equation}
  \lambda(y_i) \triangleq \log\left(\frac{\rm d\mu}{\rm d\nu}(y_i)\right) ,
  \label{eq:loglikelihoodratio}
\end{equation}
where $d\mu/d\nu$ is the Radon-Nikodym derivative.

We define $f_k:\mathbb R^{mk}\rightarrow[0,\,1]$, the detector at time $k$, as a mapping from the measurement space $\mathbf{Y}(k)$ to the interval $[0,\,1]$.
When $f_k(\mathbf{Y}(k)) = 0$, the system makes a decision $\hat{\theta} = 0$, and when $f_k(\mathbf{Y}(k))=1$, $\hat{\theta} = 1$.
When $f_k(\mathbf{Y}(k)) = \gamma\in (0,\,1)$, the system then ``flips a biased coin'' to choose $\hat{\theta} = 1$ with probability $\gamma$ and $\hat{\theta} = 0$ with probability $1-\gamma$.
The system's strategy $f \triangleq (f_1,\,f_2,\,\cdots)$ is defined as an infinite sequence of detectors from time $1$ to $\infty$.

\subsection{Attack Model} \label{subsec:attackmodel}
{
Let the (\emph{manipulated}) measurements received by the fusion center at time $k$ be
\begin{equation} \label{eqn:biasinjection}
  \mathbf{y}'(k)= \mathbf y(k)+ \mathbf{y}^a(k),
\end{equation}
where $\mathbf{y}^a(k)\in\R^m$ is the bias vector injected by the attacker at time $k$. 
In the following, Assumptions~\ref{assumpt:sparse}--\ref{assumpt:measurement} are made on the attacker, among which Assumption~\ref{assumpt:sparse} is in essence the only limitation we pose.
\begin{assumption}[Spare Attack] \label{assumpt:sparse}
There exists an index set $\mathcal I\subset \M \triangleq\{1,2,\ldots,m\}$ with $|\I| = n$ such that  $\bigcup _{k=1}^{\infty}\supp(\mathbf{y}^a(k)) =\I$.  Furthermore, the system knows the number $n$, but it does not know the set $\mathcal I$.
\end{assumption}
We should remark that the above assumption does not pose any restrictions on the value of $y_i^a(k)$ if sensor $i$ is compromised at time $k$, i.e., the bias injected into the data of a compromised sensor can be arbitrary.

Assumption~\ref{assumpt:sparse} says that the attacker can compromise up to $n$ out of $m$ sensors at each time. It is practical to assume that the attacker possesses limited resources, i.e., the number of compromised sensors is (non-trivially) upper bounded, since otherwise it would be too pessimistic and the problem becomes trivial.
The quantity $n$ might be determined by the \textit{a priori} knowledge about the quality of each sensor.
Alternatively, the quantity $n$ may be viewed as a design parameter, which indicates the resilience level  that the system is willing to introduce;  the details of which are in Remark~\ref{remark:n}.
Notice also that since the worst-case attacks (over the set of compromised sensors and the attack strategy) are concerned (the performance metric will be introduced shortly), it is equivalent to replace the cardinality requirement $|\I| = n$ by $|\I|\leq n$.
We should note that in~\cite{kailkhura2015distributed, abdelhakim2014distributed, vamvoudakis2014detection}, it was also assumed that the number/fraction of malicious sensor nodes is known to the system operator.

Moreover, the set of compromised sensors is assumed to be fixed over time. Notice that if we assume that the set of compromised sensors has a fixed cardinality but is time-varying, i.e., there exists no a set like $\I$ to bound the compromised sensors, the attacker would be required to abandon the sensor nodes it has compromised, which is not sensible. Notice that in~\cite{abdelhakim2014distributed, soltanmohammadi2013decentralized,   soltanmohammadi2014fast}, it was assumed the set of malicious/misbehaving sensors is fixed as well. We should also note that though this work is concerned with asymptotic performances (i.e., the security and efficiency introduced later),  the numerical simulations in Section~\ref{sec:simulation} show that our algorithm indeed perform quite well in an non-asymptotic setup. Actually, if a finite-time horizon problem is considered and the time required for an attacker to control a benign sensor is large enough, then it is reasonable to assume that the set of compromised sensors is fixed.

In fact, the exactly same sparse attack model as in Assumption~\ref{assumpt:sparse} has been widely adopted  by literature dealing with Byzantine sensors, e.g., state estimation~\cite{fawzi2014secure,mishra2017secure}, and quickest change detection~\cite{fellouris2017efficient}.

Finally, we should note that we do not assume any pattern of the bias $y_i^a(k)$ for $i\in\I$, i.e., the malicious bias injected may be correlated across the compromised sensors and correlated over time. Compared to the independence assumption in~\cite{marano2009distributed,kailkhura2015distributed, kailkhura2015asymptotic, abdelhakim2014distributed, soltanmohammadi2013decentralized,   soltanmohammadi2014fast, abrardo2016game}, our assumption improves the effectiveness of the attacker and is more realistic in the sense that the attacker is malicious and will do whatever it wants.

\begin{remark}~\label{remark:n}
  The parameter $n$ can also be interpreted as how many bad sensors the system can and is willing to tolerate, which is a design parameter for the system operator.
  In general, increasing $n$ will increase the resilience of the detector under attack.
  However, as is shown in the rest of the paper, a larger $n$ may result in more conservative design and is likely to cause a performance degradation during normal operation when no sensor is compromised.
\end{remark}

\begin{assumption}[Model Knowledge]
The attacker knows the probability measure $\mu$ and $\nu$ and the true state $\theta$.
\end{assumption}
By the knowledge about the sensor, the attacker can develop the probability measure $\mu$ and $\nu$. To obtain the true state, the attacker may deploy its own sensor network. Though it might be difficult to satisfy in practice, this assumption is in fact conventional in literature concerning the worst-case attacks, e.g.,~\cite{marano2009distributed, fellouris2017efficient}. Nevertheless, this assumption is in accordance with the Kerckhoffs's principle.

\begin{assumption}[Measurement Knowledge]  \label{assumpt:measurement}
At time $k$, the attacker knows the current and all the historical measurements available at the compromised sensors.
\end{assumption}
Since the attacker knows the true measurement of a compromised sensor $i$, $y_i(k)$, it may set the fake measurement arrived at the fusion center $y'_i(k)$ to any value it wants by injecting $y_i^a(k)$.
One may also verify that all the results in this paper remain even if the attacker is ``strong'' enough where at time $k$, it knows measurements from all the sensors  $\mathbf Y(k)$.

An admissible attack strategy is any causal mapping from the attacker's available information to a bias vector that satisfies Assumption~\ref{assumpt:sparse}. This is formalized as follows.
%
Let $\I = \{i_1,i_2,\ldots,i_n\}$.
Define the true measurements of the compromised sensors from time $1$ to $k$ as
\begin{equation*}
 \mathbf Y_\mathcal I(k) \triangleq \begin{bmatrix}
   \mathbf  y_\mathcal I(1)& \mathbf y_\mathcal I(2)&\cdots& \mathbf y_\mathcal I(k)
  \end{bmatrix}\in\mathbb R^{|\mathcal I|k}
\end{equation*}
with
\begin{equation*}
  \mathbf y_\mathcal I(k) \triangleq \begin{bmatrix}
    y_{i_1}(k)& y_{i_2}(k)&\cdots& y_{i_n}(k)
  \end{bmatrix}\in\mathbb R^{|\mathcal I|}.
\end{equation*}
Similar to $\mathbf Y(k)$, $\mathbf Y'(k)$ ($\mathbf{Y}^a(k)$) is defined as all the manipulated (bias vector) from time $1$ to $k$.
The bias $\mathbf y^a(k)$ is chosen as a function of the attacker's available information at time $k$:
\begin{align} \label{eqn:attackdefine}
  \mathbf y^a(k) \triangleq g(\mathbf Y_{\mathcal I}(k), \mathbf Y^a(k-1), \mathcal I,\theta,k),
\end{align}
where $g$ is a function\footnote{The function $g$ is possibly random. For example, given the available information, the adversary can flip a coin to decide whether to change the measurement or not.} of ${\mathbf Y_{\I}(k)}, \mathbf Y^a(k-1), \mathcal I,\theta,k$ such that $\mathbf y^a(k)$ satisfies Assumption~\ref{assumpt:sparse}.
We denote $g$ as an admissible attacker's strategy. Notice that since time $k$ is an input variable and the available measurements $\mathbf Y_{\I}(k), \mathbf Y^a(k-1)$ are ``increasing" with respect to time $k$, the definition in~\eqref{eqn:attackdefine} does not exclude the time-varying attack strategy.
Denote the probability space generated by all manipulated measurements $\mathbf y'(1),\,\mathbf y'(2),\,\dots$ as $(\Omega,\,\mathcal F,\,\mathbb P_{\theta})$.
The expectation taken with respect to the probability measure $\mathbb P_{\theta}$ is denoted by $\mathbb E_{\theta}$.

}

\subsection{Asymptotic Detection Performance}
Given the strategy of the system and the attacker, the probability of error at time $k$ can be defined as
\begin{equation}  \label{eqn:errorProb}
  e(\theta,\mathcal I,k) \triangleq \begin{cases}
    \mathbb E_0 f_k(\mathbf Y'(k)) & \text{when }\theta = 0,\\
    1-\mathbb E_1 f_k(\mathbf Y'(k)) & \text{when }\theta = 1.
  \end{cases}
\end{equation}
Notice that $f_k$ could take any value from $[0,1]$.
Hence, the expected value of $f_k$ is used to compute the probability of error.
In this paper, we are concerned with the worst-case scenario.
As a result, let us define
\begin{align} \label{eqn:worstcase}
  \epsilon(k) \triangleq \max_{\theta=0,1,|\mathcal I|=n} e(\theta,\mathcal I,k).
\end{align}
In other words, $\epsilon(k)$ indicates the worst-case probability of error considering all possible sets of compromised sensors and the state $\theta$ {given the detection rule $f$ and attack strategy $g$.
Notice also that in accordance with Assumption~\ref{assumpt:sparse}, the set $\I$ in the above equation is fixed over time.}

Ideally, for each time $k$, the system wants to design a detector $f_k$ to minimize $\epsilon(k)$.
However, such a task can hardly be accomplished analytically since the computation of the probability of error usually involves numerical integration.
Thus, in this article, we consider the asymptotic detection performance in hope to provide more insight on the detector design.
Define the rate function as
\begin{equation} \label{eqn:RateObj}
  \rho \triangleq \liminf_{k\rightarrow\infty} - \frac{\log \epsilon(k)}{k}.
\end{equation}

Clearly, $\rho$ is a function of both the system strategy $f$ and the attacker's strategy $g$.
As such, we will write $\rho$ as $\rho(f,g)$ to indicate such relations.
Since $\rho$ indicates the rate with which the probability of error goes to zero, the system would like to maximize $\rho$ in order to minimize the detection error.
On the contrary, the attacker wants to decrease $\rho$ to increase the detection error.

\subsection{Interested Problems}
In practice, {the attacker may not be present consistently}. As a result, the system will operate for an extended period of time with all sensors being benign.
Thus, a natural question arises: is there any detection rule that has ``decent'' performance regardless of the presence of the attacker?
Or is there a fundamental trade-off between security and efficiency?
In other words, a detector that is ``good'' in the presence of an adversary will be ``bad'' in a benign environment.
This paper is devoted to answering this question.

Informally, the performance of a detection rule when there is no attacker at all is referred to by ``efficiency'', while the performance when the worst-case attacker (provided that the attacker knows the detection rule used by the system) is present is referred to by ``security''.
Mathematically speaking, given a system strategy~$f$, denote by $\mathpzc{E}(f)$ and $\mathpzc{S}(f)$ its efficiency and security respectively, which are formalized as follows:
\begin{align}
  \mathpzc{E}(f) \triangleq & \rho(f,g= \mathbf 0),          \\
  \mathpzc{S}(f) \triangleq & \inf_g\rho(f,g)
\end{align}
where $\mathbf 0\in\R^m$ is the zero vector.


\section{Preliminary: Large Deviation Theory }
\label{sec:prilimilary}
In this section, we introduce the large deviation theory, which is a key supporting technique of this paper.

To proceed, we first introduce some definitions.
Let $M_\omega(w)\triangleq \int_{\mathbb{R}^d} e^{w\cdot X}d\omega(X), w\in\mathbb{R}^d$ be the moment generating function for the random vector $X\in\mathbb{R}^d$ that has the probability measure~$\omega$, where $w\cdot X$ is the dot product.
Let $\dom_\omega \triangleq \{w\in\mathbb{R}^d|M_\omega(w)<\infty\}$ be the support such that $M_\omega(w)$ is finite.
Define the Fenchel--Legendre transform of the function $\log M_\omega(w)$ as
\begin{align} \label{eq:ratefunc}
  I_\omega(x) = \sup_{w\in\mathbb{R}^d} \{x\cdot w  - \log M_\omega(w)\},\: x\in\mathbb{R}^d.
\end{align}
\begin{theorem}[Multidimensional Cram\'{e}r's Theorem~\cite{dembo2009large}]
  Suppose $X(1),\ldots,X(k),\ldots $ be a sequence of i.i.d. random vectors and $X(k)\in\mathbb{R}^d$ has the probability measure $\omega$.
  Let $\overline{X}(k) \triangleq \sum_{t=1}^kX(t)/k, k\in\mathbb{Z}_{+}$ be the empirical mean.
  Then if $0\in\inte(\dom_{\omega})$,
  the probability $\mathbb{P}(\overline{X}(k)\in \mathcal{A})$ with $\mathcal{A}\subseteq\mathbb{R}^d$ satisfies the large deviation principle, i.e.,
\begin{enumerate}
\item if $\mathcal{A}$ is closed,
  \begin{align*}
    \limsup_{k\to\infty} \frac{1}{k} \log \mathbb{P}(\overline{X}(k) \in \mathcal{A})  \leq -\inf_{x\in\mathcal{A}} I_\omega(x).
  \end{align*}

\item if $\mathcal{A}$ is open,
  \begin{align*}
    \liminf_{k\to\infty} \frac{1}{k} \log \mathbb{P}(\overline{X}(k) \in \mathcal{A})  \geq -\inf_{x\in\mathcal{A}} I_\omega(x).
  \end{align*}
\end{enumerate}
\end{theorem}
%
%

\section{Main Results}
\label{sec:fundamentalLimitation}

\subsection{Technical Preliminaries}

Denote the moment generating function of the log-likelihood ratio $\lambda$ under each hypothesis as:
\begin{align}
  M_0(w) &\triangleq \int_{y=-\infty}^\infty \exp(w \lambda(y)){\rm d}\nu(y),\\
  M_1(w) &\triangleq \int_{y=-\infty}^\infty \exp(w \lambda(y)){\rm d}\mu(y).
\end{align}
Furthermore, define $\dom_0$ as the region where $M_0(w)$ is finite and $I_0(x)$ as the Fenchel--Legendre transform of $\log M_0(w)$.
The quantities $M_1(w)$, $\dom_1$ and $I_1(x)$ are defined similarly.

Denote the the Kullback-Leibler (K--L) divergences by
\[
  D(1\|0)\triangleq \int_{y=-\infty}^\infty \lambda(y) {\rm d}\mu,\, D(0\|1)\triangleq -\int_{y=-\infty}^\infty \lambda(y) {\rm d}\nu.
\]
To apply the multidimensional Cram\'{e}r's Theorem and avoid degenerate problems, we adopt the following assumptions:
\begin{assumption}
  $0\in\inte(\dom_0)$ and $0\in\inte(\dom_1)$.
\label{assumpt:0interior}
\end{assumption}

\begin{assumption}
  The K--L divergences are well-defined, i.e., $0<D(1\|0),D(0\|1)<\infty$.
 \label{assumpt:finiteKLs}
\end{assumption}
With the above assumptions, we have the following properties of $I_0(x)$ and $I_1(x)$. the proof of which is provided in Appendix~\ref{Proof:ProperRateFunc}.
\begin{theorem} \label{theorem:ProperRateFunc}
  Under Assumptions~\ref{assumpt:0interior}~and~\ref{assumpt:finiteKLs}, the followings hold:

  \begin{enumerate}
    \item $I_0(x)$ ($I_1(x)$) is twice differentiable, strictly convex and strictly increasing (strictly decreasing) on $[-D(0\|1),D(1\|0)]$.
    \item The following equalities hold:
      \begin{align}
	I_1(D(1\|0)) &=0, \label{eqn:RateFuncKL00}\\
	I_0(D(1\|0)) &= D(1\|0),  \label{eqn:RateFuncKL0}  \\
	I_0(-D(0\|1)) &= 0,\label{eqn:RateFuncKL10}\\
	I_1(-D(0\|1)) &= D(0\|1).  \label{eqn:RateFuncKL1}\\
	I_0(0)&=I_1(0).\label{eqn:I00equalI10}
      \end{align}
  \end{enumerate}
\end{theorem}
Since $I_0(0) = I_1(0)$, let us define
\begin{align} \label{eqn:DefC}
  C\triangleq I_0(0).
\end{align}
To make the presentation clear, we illustrate $I_0(x)$ and $I_1(x)$ in Fig.~\ref{Fig:RateFunc}.

The ``inverse functions'' of $I_0(x)$ and $I_1(x)$ are defined as follows: for $z\geq 0$,
\begin{align*}
  I^{-1}_0(z) = &  \max\{x\in\mathbb{R}:I_0(x)=z\},\\
  I^{-1}_1(z) = &  \min\{x\in\mathbb{R}:I_1(x)=z\}.
\end{align*}
Let  $D_{\min}\triangleq \min\{D(0\|1),D(1\|0)\}$.
We further define $h(z):\left(0,(m-n)D_{\min}\right) \mapsto \left(0,(m-n)D_{\min}\right)$  as
\begin{align*}
  &h(z)\\
  \triangleq&
  (m-n)\min\{I_0(I^{-1}_1(z/(m-n)), I_1(I^{-1}_0(z/(m-n))\}. \addtag
  \label{eqn:definitionhz}
\end{align*}

\begin{center}
\begin{figure}
\center
\input{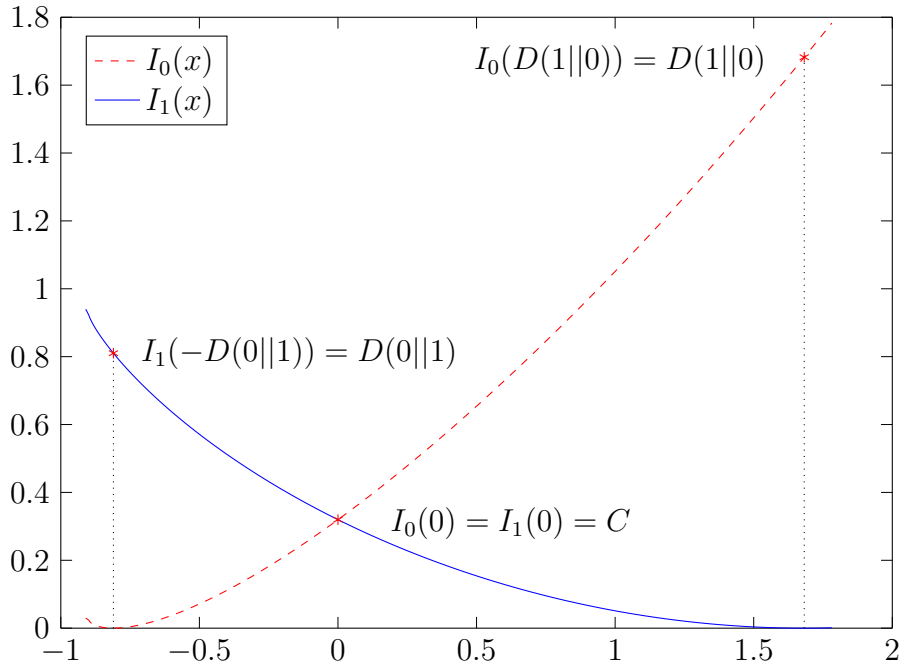}

\caption{Illustration of $I_0(x)$ and $I_1(x)$. The figure is plotted by assuming $y_1(1)$ to be Bernoulli distributed under both hypotheses with $\mathbb P_0^{o}(y_1(1)=1)=0.02$ and $\mathbb P_1^{o}(y_1(1)=1)=0.6$.}
\label{Fig:RateFunc}
\end{figure}
\end{center}

\subsection{Fundamental Limits}

We are ready to provide the fundamental limitations between efficiency and security.
The proof is provided in Appendix~\ref{Proof:FundamentalLimit}.
\begin{theorem} \label{theorem:PerformanceLimitation}
  For any detection rule $f$, the following statements on $\E(f)$ and $\bS(f)$ are true:
  \begin{enumerate}
  \item $\E(f)\leq mC$,
  \item $\bS(f)\leq(m-2n)^{+}C$, where $(m-2n)^+ = \max\{0,m-2n\}$.
  \item $\bS(f)\leq \E(f)$,
  \item Let $\E(f)=z$, we have
    \begin{subnumcases}{\bS(f)\leq}
      h(z) &if  $0<z<(m-n)D_{\min}$ \label{E-S01}
      \\
      0 &if  $z\geq (m-n)D_{\min}$. \label{E-S02}
    \end{subnumcases}
  \end{enumerate}
\end{theorem}

\begin{remark}
  Theorem~\ref{theorem:PerformanceLimitation} indicates that $mC$ is the maximum efficiency that can be achieved by any detector, while $(m-2n)^+C$ is the maximum security that can be achieved. Therefore, if $m \leq 2n$, i.e., no less that half of the sensors are compromised, then $\bS(f) = 0$ for any $f$, which implies that all detectors will have zero security. In that case, the naive Bayes detector will be the optimal choice since it has the optimal efficiency and the analysis becomes trivial. Therefore, without any further notice, we assume $m>2n$ for the rest of the paper.

  Notice that fourth constraint in Theorem~\ref{theorem:PerformanceLimitation} indicates a trade-off between security and efficiency. For general cases, the maximum security and efficiency may not be achieved simultaneously. However, in Section~\ref{sec:symmetric}, we will prove that for a special case, there exist detectors that can achieve the maximum security and efficiency at the same time.
\end{remark}

Notice that $I_0(x)$ ($I_1(x)$) is strictly increasing (decreasing) on $[-D(0\|1),D(1\|0)]$.
Therefore, combining~\eqref{eqn:RateFuncKL0}~and~\eqref{eqn:RateFuncKL1}, one obtains that $h(z)$ is strictly decreasing.
Then the dual version of~\eqref{E-S01} is obtained as follows.
Let $\bS(f)=z$ we have that  if $0<z\leq(m-2n)C$,
\begin{align} \label{S-E01}
  \E(f)\leq h^{-1}(z) = h(z),
\end{align}
where $h^{-1}(z)$ is the inverse function of $h(z)$, and the equality holds because $h(z)$ is an involutory function, i.e., $h(h(z))=z$ for every $z\in(0,(m-n)D_{\min})$.

We then have the following two corollaries.
The results follow straightforwardly from Theorem~\ref{theorem:PerformanceLimitation} and~\eqref{S-E01}, we thus omit the proofs.
\begin{corollary} \label{corollary:S-E}
  Suppose the security of a detector $f$ satisfies
\begin{align*}
  \bS(f) = z \in [0, (m-2n)C],
\end{align*}
then the maximum efficiency of $f$ satisfies the following inequality:
  \begin{align*}
    \max_{ f\in\{ f:\mathpzc{S}(f)=z \} } \mathpzc{E}(f) \leq\begin{cases}
      mC & \text{if } z=0
      \\
      h_e(z) & \text{if }  z > 0
    \end{cases},
  \end{align*}
  where
  $h_e(z)\triangleq\min\{ mC, h(z)\}$.
\end{corollary}
\begin{corollary}  \label{corollary:E-S}
  Suppose the efficiency of a detector $f$ satisfies
  \begin{align*}
    \mathpzc{E}(f) = z \in [0, mC],
  \end{align*}
then the maximum security of $f$ satisfies the following inequality:
  \begin{align*}
    \max_{f\in\{f:\mathpzc{E}(f)=z\}} \mathpzc{S}(f)
    \leq\begin{cases}
      h_s(z)& \text{if }  0<z < z', \\
      0& \text{if } z \geq z'\text{ or } z=0,
    \end{cases}
  \end{align*}
  where $h_s(z) = \min\{z,(m-2n)C,h(z)\}$, and $z' = (m-n)D_{\min}$.
\end{corollary}

%

\subsection{Achievability}    \label{sec:subAchieve}
In this section, we propose a detector that achieves the upper bounds in Corollaries~\ref{corollary:S-E}~and~\ref{corollary:E-S}.

Let $z_s\leq(m-2n)C$.
At time $k\geq 1$, the algorithm, denoted by $f_{z_s}^*$, is implemented as follows.
\begin{algorithm}
1:\: Compute the empirical mean of the likelihood ratio from time $1$ to time $k$ for each sensor $i$:
  \begin{align*}
    \bar\lambda_i(k) &\triangleq \sum_{t=1}^k \lambda(y_i'(t))/k \\
                 &= \frac{k-1}{k}\bar\lambda_i(k-1) +  \frac{1}{k}\lambda(y_i'(k)) \addtag \label{eqn:barlambda}
  \end{align*}
  with $\bar\lambda_i(0) = 0$.\\
2:\: Compute $I_0(\bar\lambda_i(k))$ and $I_1(\bar\lambda_i(k))$ for each $i$. Compute the following sum:
  \begin{align*}
    \delta(0,k) &\triangleq \min_{|\mathcal O|=m-n, \mathcal O \subset \M}\sum_{i\in\mathcal O} I_0(\bar\lambda_i(k)), \\
    \delta(1,k) &\triangleq \min_{|\mathcal O|=m-n, \mathcal O \subset \M}\sum_{i\in\mathcal O} I_1(\bar\lambda_i(k)).
  \end{align*} \\
3:\: If $\delta(0,k)<z_s$, make a decision $\hat\theta=0$; go to the next step otherwise. \\
4:\: If $\delta(1,k)<z_s$, make a decision $\hat\theta=1$; go to the next step otherwise. \\
5:\: If $\sum_{i=1}^{m} \bar\lambda_i(k) <0$, make a decision $\hat\theta=0$; make a decision $\hat\theta=1$ otherwise.\\
\caption{Hypothesis testing algorithm $f^*_{z_s}$}
\label{alg:fstar}
\end{algorithm}
\begin{remark}
  We here discuss about the computational complexity of the detection rule $f_{z_s}^*$.
  The computational complexity for the  step~1 is $O(m)$.
  Notice that the quantity $\bar\lambda_i(k)$ is computed in a recursive fashion.
  The complexity for the step~2 is $O(m\log m)$.
  To compute $\delta(0,k)$ and $\delta(1,k)$, we can first sort $I_0(\bar\lambda_i(k))$ and $I_1(\bar\lambda_i(k))$ in ascending order, respectively, and then sum the first $m-2n$ elements of each.
  The computational complexity for the step~3 and  step~4 is fixed, and the step~5 has computational complexity $O(m)$.
  Therefore, the total computational complexity for each time step is $O(m\log m)$.
\end{remark}

We now show the performance of $f_{z_s}^*$ and the proof is provided in Appendix~\ref{sec:proofAchieve}.

\begin{definition}
  $(z_e,z_s)$ are called an admissible pair if the following inequalities holds:
  \begin{align*}
   0\leq  &z_s \leq (m-2n)C, \\
    &z_e \leq\begin{cases}
      mC & \text{if } z_s=0
      \\
      h_e(z_s) & \text{if }  z_s > 0
    \end{cases},
  \end{align*}
  where $h_e(z_s)$ is defined in Corollary~\ref{corollary:S-E}.
\end{definition}

\begin{theorem} \label{theorem:achieve}
  Let $(z_e,z_s)$ be any admissible pair of efficiency and security.
  Then there holds
  \[\E(f_{z_s}^*)\geq z_e, \quad \bS(f_{z_s}^*) \geq z_s.\]
\end{theorem}

The above theorem means that the upper bounds in Corollaries~\ref{corollary:S-E}~and~\ref{corollary:E-S} are achieved by $f_{z_s}^*$.
Hence, we provide a tight characterization on admissible efficiency and security pair. We illustrate the shape of admissible region in Fig~\ref{Fig:FundaLim}.

\begin{figure}
\center

\input{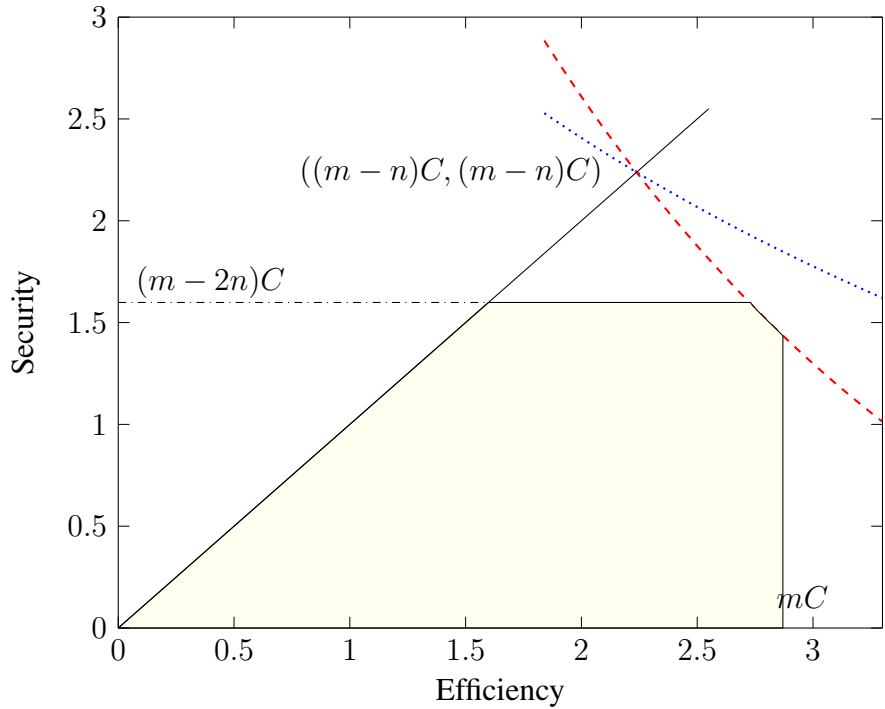}

\caption{Achievable efficiency and security region for any detector. The figure is plotted by assuming $y_1(1)$ to be Bernoulli distributed under both hypotheses with $\mathbb P_0^{o}(y_1(1)=1)=0.02$ and $\mathbb P_1^{o}(y_1(1)=1)=0.6$. The shaded area is the admissible pair $(\E(f),\bS(f))$ for any detector $f$. The red dashed line is the function $(m-n)I_0(I_1^{-1}(z/(m-n)))$, while the blue dotted line the function $(m-n)I_1(I_0^{-1}(z/(m-n)))$.}
\label{Fig:FundaLim}
\end{figure}


\begin{remark}
  The optimal detector may not be necessarily unique, in the sense that there may exist other detectors, other than the one defined by Algorithm~1, that can achieve the same efficiency and security limits.
  By definition, the detectors achieving the limits have the same asymptotic performance.
  However, the finite-time performance (in terms of detection error) may be different and we are planning to investigate this in the future work.
\end{remark}

\subsection{Special Case: Symmetric Distribution}

\label{sec:symmetric}

In this subsection, we discuss a case where the maximum security and efficiency can be achieved simultaneously by a detector.

Notice that by the definition of admissible pair, if we have
\begin{align}
  h((m-2n)C) \geq mC,
  \label{eq:hvsmC}
\end{align}
then we know that $(z_e = mC, z_s = (m-2n)C)$ is an admissible pair and hence the detector $f^*_{(m-2n)C}$ defined in Section~\ref{sec:subAchieve} can achieve maximum security $(m-2n)C$ and efficiency $mC$ simultaneously. In other words, adding security will not deteriorate the performance of the system in the absence of the adversary.

The following theorem provides a sufficient condition for \eqref{eq:hvsmC}, which is based on the first order derivative of $I_0(\cdot)$ and $I_1(\cdot)$.
The proof is presented in Appendix~\ref{Proof:SEsimultaneously} for the sake of legibility.
\begin{theorem}  \label{theorem:SEsimultaneously}
  If $I_0^{(1)}(x)|_{x=0} = - I_1^{(1)}(x)|_{x=0}$, then $h((m-2n)C) \geq mC$ holds. Therefore, $f^*_{(m-2n)C}$ possesses not only the maximal security but also the maximal efficiency.
\end{theorem}

  Notice that whether or not the above sufficient condition is satisfied merely depends on the probability distribution of the original observations, which is independent of the number of the compromised sensors.

  If there exists ``symmetry'' between distribution $\mu$ and $\nu$, then the sufficient condition can be satisfied. To be specific, if there exists a constant $a$ such that for any Borel measurable set $\mathcal A$, we have
\begin{align*}
  \mu(a+\mathcal A) = \nu (a-\mathcal A),
\end{align*}
 then one can prove that
  \begin{align*}
    M_0(w) = M_0(-w),
  \end{align*}
which further implies that
  \begin{align*}
    I_0(x) = I_1(-x)\Rightarrow I_0^{(1)}(x)|_{x=0} = - I_1^{(1)}(x)|_{x=0}.
  \end{align*}

  We provide two examples of pairs of ``symmetric'' distributions as follows:
  \begin{enumerate}
  \item Each $y_i(k)$ is i.i.d. Bernoulli distributed and
    \begin{align*}
      y_i(k) = \begin{cases}
        \theta & \text{with probability }p_0\\
        1-\theta & \text{with probability }1-p_0
        \end{cases}
    \end{align*}
  \item Each $y_i(k)$ satisfies the following equation:
    \begin{align*}
      y_i(k) = a\theta + v_i(k),
    \end{align*}
    where $a \neq 0$ and $v_i(k)\sim \mathcal N(\bar v, \sigma^2)$ is i.i.d. Gaussian distributed.
  \end{enumerate}

\section{Extension}
\label{sec:extension}
In this section, we consider two extensions to the problem settings discussed in Section~\ref{sec:fundamentalLimitation}.
\subsection{Secure Sensors}
Consider that there is a subset of ``secure'' sensors that are well protected and cannot be compromised by the attacker. We would like to study the trade-off between security and efficiency when those ``secure'' sensors are deployed.

Let $m_s$ out of the total $m$ sensors  be ``secure" and the remaining $m-m_s$ sensors are ``normal" ones that can be compromised by an adversary. In this subsection, $n$ can take any value in $\{0,1,\ldots,m-m_s\}$ and does not necessarily satisfy $2n < m$.
The other settings are the same as in Section~\ref{sec:ProblemFormulation}.
Denote by $\mathpzc{E}_s(f)$ and $\mathpzc{S}_s(f)$ the efficiency and security of a detection rule $f$ under such case.

Then one obtains the following results as in Theorem~\ref{theorem:PerformanceLimitation}.
\begin{theorem}
  For any detection rule $f$, the following statements on $\E_s(f)$ and $\bS_s(f)$ are true:
  \begin{enumerate}
  \item $\E_s(f)\leq mC$,
  \item $\bS_s(f)\leq \max((m-2n)C, m_sC)$,
  \item $\bS_s(f)\leq \E_s(f)$,
  \item Let $\E_s(f)=z$, we have
  \begin{align*}
   \bS_s(f)\leq \begin{cases}
      h(z) &\text{if }0<z<(m-n)D_{\min}
      \\
      0 &\text{if }z\geq (m-n)D_{\min}
    \end{cases}.
    \end{align*}
  \end{enumerate}
\end{theorem}
The above theorem is proved in the same manner as in Appendix~\ref{Proof:FundamentalLimit}. Notice that the essential difference is the range of $\bS_s(f)$, i.e., the statement in the second bullet. This is due to the fact that the $m_s$ secure sensors cannot be compromised.
\begin{remark}
From the above theorem, one sees that replacing $m_s$ normal sensors with secure sensors does not change the fundamental trade-off between the security and efficiency. However, the benefit of these $m_s$ secure sensors are that the security itself is improved when $2n > m-m_s$. Also, one notice that when $m-m_s\geq 2n$, there are no gains of deploying secure sensors. Intuitively, in such case the redundancy of the $m-m_s$ normal sensors is enough.
\end{remark}

Furthermore, the detector $f^s_{z_s}$ in Algorithm~\ref{alg:variation}, which is a slight variation of $f^*_{z_s}$ and treats the $m_s$ secure sensors separately, achieves the limits. This is stated in the following theorem, which is proved in the same manner as in Appendix~\ref{sec:proofAchieve}.
\begin{theorem}
If the pair $(z_e,z_s)$ satisfies
  \begin{align*}
   0\leq  &z_s \leq \max((m-2n)C, m_sC), \\
    &z_e \leq\begin{cases}
      mC & \text{if } z_s=0
      \\
      h_e(z_s) & \text{if }  z_s > 0
    \end{cases},
  \end{align*}
then there holds
  \[\E_s(f_{z_s}^s)\geq z_e, \quad \bS_s(f_{z_s}^s) \geq z_s.\]
\end{theorem}

%
\begin{algorithm}
1:\: Compute $\bar\lambda_i(k)$ for each sensor. \\
2:\: Compute $I_0(\bar\lambda_i(k))$ and $I_1(\bar\lambda_i(k))$ for each sensor. Compute the minimum sum from the ``normal'' sensors:
  \begin{align*}
    \delta(0,k) &= \min_{|\mathcal{O}|=m-m_s-n, \mathcal{O}\subset\{1,\ldots,m-m_s\}}\sum_{i\in\mathcal{O}} I_0(\bar\lambda_i(k)), \\
        \delta(1,k) &=\min_{|\mathcal{O}|=m-m_s-n,\mathcal{O}\subset\{1,\ldots,m-m_s\}}\sum_{i\in\mathcal{O}} I_1(\bar\lambda_i(k)).
  \end{align*}\\
3:\: If $\delta(0,k) + \sum_{m-m_s+1}^{i=m}I_0(\bar\lambda_i(k))  <z_s$, make a decision $\hat\theta=0$; go to the next step otherwise. \\
4:\: If $\delta(1,k) + \sum_{m-m_s+1}^{i=m}I_1(\bar\lambda_i(k))  <z_s$, make a decision $\hat\theta=1$; go to the next step otherwise. \\
5:\: If $\sum_{i=1}^{m} \bar\lambda_i(k) <0$, make a decision $\hat\theta=0$; make a decision $\hat\theta=1$ otherwise.
\caption{Hypothesis testing algorithm $f^{s}_{z_s}$ when there are secure sensors}
\label{alg:variation}
\end{algorithm}


\subsection{Unknown Number of Compromised Sensors}
In the previous section, we assume that if the system is being attacked, then $n$ sensors are compromised. However, in practice, the exact number of compromised sensors is likely to be unknown.
In this subsection, we assume that we know an estimated upper bound on the compromised sensors, denoted by $n$. Let $n_a$ denote the number of the sensors that are actually compromised. Therefore, $n_a$ may take value in $\mathcal{N}_a\triangleq\{0,1,2,\ldots n\}$\footnote{{In Section~\ref{subsec:attackmodel}, we remark that the requirement $|\I|=n$ can be equivalently replaced by $|\I|\leq n$. The implicit assumption is that the estimated upper bound $n$ is tight and the worst-case number of compromised sensors is  in indeed $n$. Therefore, $n_a$ in this section may also be interpreted as the tight upper bound of the number of actually compromised senors. }}.

Given a detector $f$, denote by $\mathpzc{D}_{n_a}(f)$ the detection performance when the number of compromised sensor is $n_a$.
Then, one has $\mathpzc{D}_{0}(f) = \E(f)$ and $\mathpzc{D}_{n}(f) = \bS(f)$.
In the following, we present the pairwise trade-off between $\mathpzc{D}_{n_a}(f)$ and $\mathpzc{D}_{n'_a}(f)$ for any $0\leq n_a,n'_a \leq n$, and propose an algorithm to achieve these performance limits.
A similar argument as in Section~\ref{sec:fundamentalLimitation} can be adopted to obtain these results, the details of which are omitted.

We define $\tilde{h}: \mathcal{N}_a \times \mathcal{N}_a \times (0,\infty) \mapsto (0,\infty)$ as
\begin{align*}
  &\tilde{h}(n_a,n_a',z) \\
  \triangleq&
              (m- \tilde{n}_a)\min\left\{\tilde{h}_0(z/(m-\tilde{n}_a)), \tilde{h}_1(z/(m-\tilde{n}_a))\right\},
\end{align*}
where $\tilde{n}_a = n_a + n_a'$, and
\begin{align*}
\tilde{h}_0(z) &=  \begin{cases}
        I_0(I^{-1}_1(z)) & \text{ if } 0<z<D(0||1)\\
        0 & \text{if }z \geq D(0||1)
        \end{cases}, \\
\tilde{h}_1(z) &=  \begin{cases}
        I_1(I^{-1}_0(z)) & \text{ if } 0<z<D(1||0)\\
        0 & \text{if }z \geq D(1||0)
        \end{cases}.
\end{align*}
Then one obtains that for any detector $f$ and $n_a,n_a'\in\mathbb{N}$, there hold
\begin{align}
  \mathpzc{D}_{n_a}(f) &\leq (m-2n_a)C,  \label{eqn:extbound}            \\
  \mathpzc{D}_{n_a}(f) &\leq \tilde{h}\left(n_a,n_a', \mathpzc{D}_{n_a'}(f)\right). \label{eqn:exttradeoff}
\end{align}
Let $\mathbf{z}\triangleq(z_0,z_1,\ldots,z_n)$ be a $n$-tuplet of admissible detection performance, i.e.,
\begin{align*}
  z_{n_a} &\leq (m-2n_a)C,          \\
  z_{n_a} &\leq \tilde{h}(n_a,n_a', z_{n_a'}).
\end{align*}
Then the detector in Algorithm~\ref{alg:fextension}, which is a variation of $f_{z_s}^*$ in Section~\ref{sec:subAchieve} and is denoted by $f^*_{\bf z}$, can achieve these performance, i.e., $\mathpzc{D}_{n_a}(f^*_{\bf z}) \geq z_{n_a}$ for any $n_a\in\mathbb{N}$.

\begin{algorithm}
 {\bf initialization:} $n_a =n$. \\
 1:\: Compute $\bar\lambda_i(k), I_0(\bar\lambda_i(k)), I_1(\bar\lambda_i(k))$ for each sensor $i$.\\
 2:\: {\bf While} $n_a \geq 1$
\hspace*{3mm} \begin{enumerate}
   \item Compute these two minima:
  \begin{align*}
    \tilde{\delta}(0,k,n_a) &\triangleq \min_{|\mathcal O|=m-n_a, \mathcal O \subset \M}\sum_{i\in\mathcal O} I_0(\bar\lambda_i(k)), \\
    \tilde{\delta}(1,k,n_a) &\triangleq \min_{|\mathcal O|=m-n_a, \mathcal O \subset \M}\sum_{i\in\mathcal O} I_1(\bar\lambda_i(k)).
  \end{align*}
   \item If $\tilde\delta(0,k,n_a)<z_{n_a}$, make a decision $\hat\theta=0$ and stop.
   \item If $\tilde\delta(1,k,n_a)<z_{n_a}$, make a decision $\hat\theta=1$ and stop.
   \item Replace $n_a$ with $n_a -1$.
 \end{enumerate}
3:\: If $\sum_{i=1}^{m} \bar\lambda_i(k) <0$, make a decision $\hat\theta=0$; make a decision $\hat\theta=1$ otherwise.
\caption{Hypothesis testing algorithm $f^*_{\bf z}$}
\label{alg:fextension}
\end{algorithm}

%


\section{Numerical Examples}
\label{sec:simulation}
{\subsection{Asymptotic Performance}}
We simulate the performance of the detector $f_{z_s}^*$ proposed in Section~\ref{sec:subAchieve} {(i.e., its efficiency and security)} and compare the  empirical results to the theoretical ones shown in Fig.~\ref{Fig:FundaLim}.

The same parameters as in Fig.~\ref{Fig:FundaLim} are used, i.e., $\mathbb P_0^{o}(y_1(1)=1)=0.02$, $\mathbb P_1^{o}(y_1(1)=1)=0.6$, $m=9$ and $n=2$. To simulate the security, it is assumed that the following attack strategy is adopted. If $\theta=0$, the attacker modifies the observations of the compromised sensors such that for every $i\in\I$ and $k\geq 1$
\begin{align*}
I_0(\bar\lambda_i(k)) \geq z_s.
\end{align*}
On the other hand,  if $\theta=1$, the attack strategy is such that
$
I_1(\bar\lambda_i(k)) \geq z_s
$ holds for every $i\in\I$ and $k\geq 1$.

To simulate the performance with high accuracy, we adopt the importance sampling approach~\cite{rubinstein2016simulation}. To plot Fig.~\ref{Fig:Compare}, we let $z_s$ be in $[0,(m-2n)C=1.5987]$. Notice that the theoretical performance of $f^*_{z_s}$ coincides exactly with the fundamental limits in Fig.~\ref{Fig:FundaLim}. Therefore, Fig.~\ref{Fig:Compare} verifies that our algorithm $f^*_{z_s}$ indeed achieves the fundamental limits.

\begin{figure}[ht]
\center
\input{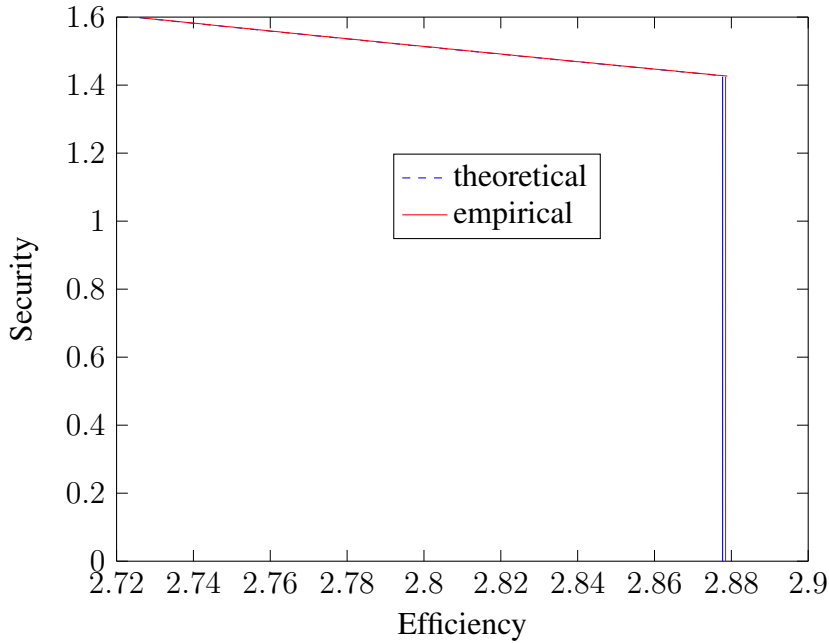}
\caption{Comparison between the empirical and  theoretical performance of the detector $f^*_{z_s}$ when $z_s\in[0,(m-2n)C]$. }
\label{Fig:Compare}
\end{figure}

{\subsection{Non-asymptotic Performance}}
{We have proved that our algorithm is optimal in the sense that it achieves the fundamental trade-off between the security and efficiency. However, notice that both the security and efficiency are asymptotic performance metrics. In this example, we show that our algorithm possesses quite ``nice" finite-time performance as well by comparing it to the naive Bayes detector. We should remark that while the Bayes detector is strictly optimal (i.e., optimal for any time horizon) in the absence of attackers, its security is zero. }
The results are in Fig.~\ref{Fig:finitetime}, where $z_s$ is chosen to be $1.4282$.
Fig.~\ref{Fig:finitetime} illustrates that the algorithm $f^*_{z_s}$ with $z_s=1.4282$ has a finite-time detection performance comparable to that of naive Bayes detector when the attacker is absent.
{The finite-time performance metric $\epsilon(k)$ is defined in~\eqref{eqn:worstcase}, where the attacker is absent, i.e., $g=0$, and the detector is $f^*_{z_s=1.4282}$ or the naive Bayes.}
One should note that the  security of $f^*_{z_s}$ is $1.4282$. As a result, adopting the secure detector $f_{z_s}^*$ increases the security of the system while introducing minimum performance loss in the absence of the adversary.

\begin{figure}[ht]
	\center
%
%
\begin{tikzpicture}

\begin{axis}[%
width=4.1in,
height=2.85in,
at={(0.758in,0.481in)},
scale only axis,
separate axis lines,
every outer x axis line/.append style={black},
every x tick label/.append style={font=\color{black}},
every x tick/.append style={black},
xmin=1,
xmax=10,
xlabel={time $k$},
every outer y axis line/.append style={black},
every y tick label/.append style={font=\color{black}},
every y tick/.append style={black},
ylabel={$\epsilon(k)$},
ymode=log,
axis background/.style={fill=white},
legend style={legend cell align=left, align=left, draw=black}
]
\addplot [color=blue,dashed]
  table[row sep=crcr]{%
1	1.3078e-02\\
2	3.9211e-04\\
3	1.3203e-05\\
4	7.9448e-07\\
5	5.0116e-08\\
6	3.2024e-09\\
7	2.0816e-10\\
8	1.3572e-11\\
9	4.9711e-13\\
10	1.9296e-14\\
};
\addlegendentry{Naive Bayes}
\addplot [color=red]
  table[row sep=crcr]{%
1	2.5074e-02\\
2	8.3486e-04\\
3	1.2897e-05\\
4	1.5836e-06\\
5	4.8561e-08\\
6	3.8571e-09\\
7	2.0210e-10\\
8	1.0938e-11\\
9	5.6596e-13\\
10	2.0986e-14\\
};
\addlegendentry{$f_{1.4282}^*$}

\end{axis}
\end{tikzpicture}%
\caption{Finite-time performance of $f^{*}_{z_s}$ in the absence of the adversary. }
\label{Fig:finitetime}
\end{figure}
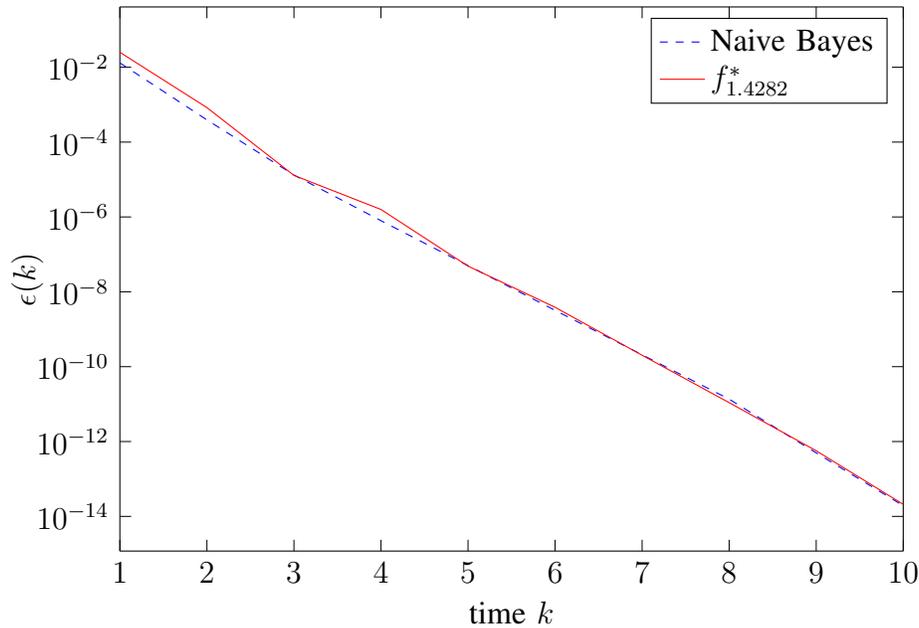

{
\subsection{Comparison with Other Detectors}
To simulate the detectors introduced later, we use the same sensor network model as in Fig.~\ref{Fig:Compare}. The asymptotic performances, i.e., the efficiency and security, are summarized in Table~\ref{table:asymptotic}, while the non-asymptotic performances when the attacker is absent are in Fig.~\ref{Fig:finiteCompare}. Table~\ref{table:asymptotic} is consistent with the statement that our algorithm achieves the best trade-off between the security and efficiency, while Fig.~\ref{Fig:finiteCompare} shows that it is preferable to adopt our algorithm as well with respect to the finite-time performance when the attacker is absent. In the following, we present the two detectors to be compared to in detail.

The first detector is the equilibrium detection rule that is proposed in~\cite{Jiaqi_CDC_detection} for cases where $m>2n$.  This detection rule, which shares the same spirit with the $\alpha$-trimmed mean in robust statistics~\cite{huber2011robust}, first removes the largest $n$ and smallest $n$ log-likelihood ratios, and then compares the mean of the remaining $m-2n$ log-likelihood ratios to $0$, just as in the classic probability ratio test. The details of the detection rule, denoted by $f_{{\rm trim}}$, are formalized as follows.
\begin{align*}
f_{{\rm trim}}(\mathbf Y'(k)) = \begin{cases}
    0 & \text{if } \sum_{i=n+1}^{m-n} \bar\lambda_{[i]}(k) <0,\\
    1 & \text{if } \sum_{i=n+1}^{m-n} \bar\lambda_{[i]}(k) \geq 0,
  \end{cases}
\end{align*}
where $\bar\lambda_{[i]}(k)$ is the $i$-th smallest element of $\{\bar\lambda_{1}(k),\bar\lambda_{2}(k),\ldots,\bar\lambda_{m}(k)\}$ with $\bar\lambda_{i}(k)$ being the empirical mean of log-likelihood ratio from time $1$ to $k$ for senor $i$, which is defined in~\eqref{eqn:barlambda}.

It was shown in~\cite{Jiaqi_CDC_detection} that the  security and efficiency of $f_{{\rm trim}}$ are
\begin{align*}
\bS(f_{{\rm trim}}) = (m-2n)C, \quad  \E(f_{{\rm trim}}) = (m-n)C.
\end{align*}
Since for any $z<C$, there hold
\begin{align*}
 I_0(I_1^{-1}(z)) > C, \quad   I_1(I_0^{-1}(z)) > C.
\end{align*}
Then by the definition of $h(z)$ and Theorem~\ref{theorem:achieve}, one obtains that if the security of our algorithm is $(m-2n)C$, its efficiency is larger than $(m-n)C$, i.e.,
\[\E(f^*_{z_s = (m-2n)C}) > (m-n)C.\]
Therefore, our algorithm is preferable since, with the same security, it achieves the larger efficiency than the algorithm $f_{{\rm trim}}$. In particular, by Theorem~\ref{theorem:SEsimultaneously}, the efficiency gain of our algorithm in certain cases is $nC$.

The next detector is the $q$-out-of-$m$ procedure~\cite{viswanathan1989counting}, which has been studied in~\cite{abdelhakim2014distributed} by assuming that the malicious sensor nodes generate fictitious data randomly and independently, and the probability that the compromised sensor flips the binary message  is known.
The $q$-out-of-$m$ procedure is simple and works as follows.  At time $k$, after receiving the $mk$ (binary) messages, the fusion center makes a decision
\begin{align}\label{eqn:qoutm}
  \hat{\theta}=\begin{cases}
      1 & \text{if } \sum_{t=1}^k\sum_{i=1}^{m} y'_i(k) \geq q_k,
      \\
      0 & \text{otherwise.}
    \end{cases}
\end{align}
Let ${\bf q} = [q_1,\ldots,q_k,\ldots]$ be a sequence of thresholds used in the above detector from time $1$ to infinity. In the sequel, we denote the above detector as $f_{{\rm qom}}(\bf q)$.
Notice that $f_{{\rm qom}}(\bf q)$ is just the naive Bayesian detector, which minimizes the weighted sum of miss detection and false alarm at each time $k$ (the weight is determined by $q_k$). It is clear that if $f_{{\rm qom}}(\bf q)$ is used at the fusion center, the worst-case attack is always sending $0$ ($1$) if the true state $\theta$ is $1$ ($0$). Therefore, at time $k$, the performance (i.e., the probability of detection error) of the detector $f_{{\rm qom}}(\bf q)$ under the worst-case attacks is as follows.
\begin{align*}
 &\mathbb P_1(f_{{\rm qom}}({\bf q}) = 0) = \sum_{j=0}^{q_k}
\begin{pmatrix}
  mk-nk \\
  j
\end{pmatrix} p_1^j(1-p_1)^{mk-nk-j}, \\
 &\mathbb P_0(f_{{\rm qom}}({\bf q}) = 1)\\
  &= \sum_{j=\max(0,q_k-nk)}^{mk-nk}
\begin{pmatrix}
  mk-nk \\
  j
\end{pmatrix} p_0^j(1-p_0)^{mk-nk-j},
\end{align*}
where $p_0\triangleq \mathbb P_0^{o}(y_1(1)=1) = 0.02, p_1\triangleq \mathbb P_1^{o}(y_1(1)=1) =0.6$.
Then it is reasonable to set $nk<q_k<mk-nk$, since otherwise the worst-case (over $
\theta$) detection error will be $1$. However it is challenging to obtain the optimal $q_k$ analytically to minimize the worst-case detection error; we do this by brute-force numerical simulations. By varying the time $k$ from $1$ to $40$, we obtain the (approximate) security and the optimal parameters $q_1^*,\ldots,q_{40}^*$. Then we further simulate the performance of the $q$-out-$m$ algorithm when the optimal parameters obtained above are used and the attacker is absent.

\begin{table}[!htb]
\center
  \caption{The asymptotic performances of our algorithm $f_{z_s}^*$ with $z_s = 1.4282$, the trimmed mean detector in~\cite{Jiaqi_CDC_detection} $f_{\rm trim}$,  and the optimal q-out-m procedure $f_{{\rm qom}}(\bf q^*)$. }
  \begin{tabular}{ |c || c | c | c |}
    \hline
     & $f_{z_s = 1.4282}^*$ & $f_{\rm trim}$ & $f_{{\rm qom}}(\bf q^*)$ \\ \hline

     security & 1.43 & 1.43 & 0.69    \\ \hline
     efficiency& 2.88  & 2.00 & 1.68  \\ \hline

  \end{tabular}
 \label{table:asymptotic}
\end{table}

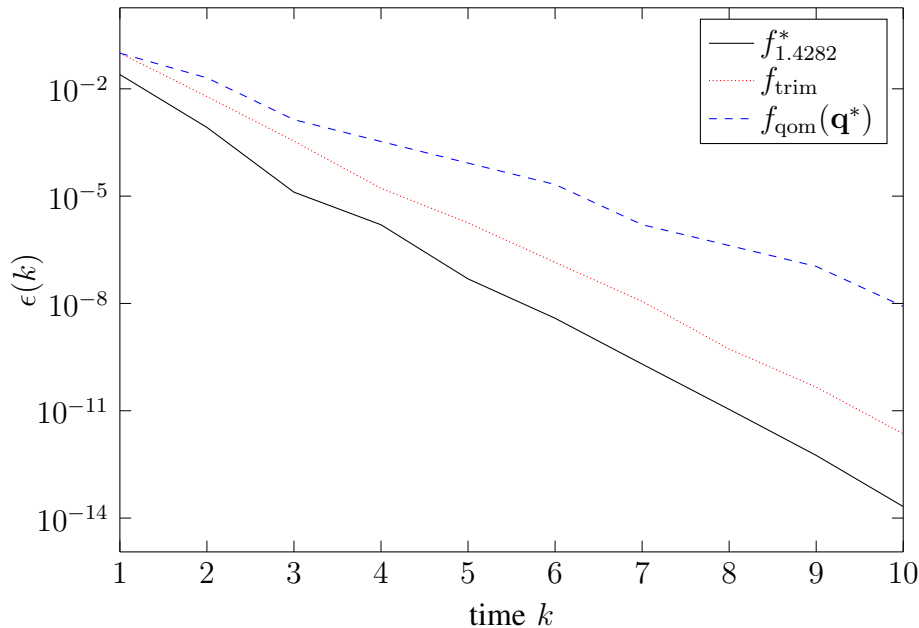
\begin{figure}[ht]
	\center
%
%
%
\begin{tikzpicture}

\begin{axis}[%
width=4.1in,
height=2.85in,
at={(0.758in,0.481in)},
scale only axis,
separate axis lines,
every outer x axis line/.append style={black},
every x tick label/.append style={font=\color{black}},
every x tick/.append style={black},
xmin=1,
xmax=10,
xlabel={time $k$},
every outer y axis line/.append style={black},
every y tick label/.append style={font=\color{black}},
every y tick/.append style={black},
ylabel={$\epsilon(k)$},
ymode=log,
axis background/.style={fill=white},
legend style={legend cell align=left, align=left, draw=black}
]

\addplot [color=black]
  table[row sep=crcr]{%
1	2.5074e-02\\
2	8.3486e-04\\
3	1.2897e-05\\
4	1.5836e-06\\
5	4.8561e-08\\
6	3.8571e-09\\
7	2.0210e-10\\
8	1.0938e-11\\
9	5.6596e-13\\
10	2.0986e-14\\
};
\addlegendentry{$f_{1.4282}^*$}

\addplot [color=red,densely dotted]
  table[row sep=crcr]{%
1	0.0994537\\
2	0.00609034\\
3	0.000345445 \\
4	1.64921e-5\\
5	1.78819e-6\\
6	1.40509e-7\\
7	1.14828e-8\\
8	5.32103e-10 \\
9	4.57181e-11\\
10	2.2707e-12\\
};
\addlegendentry{$f_{{\rm trim}}$}

\addplot [color=blue,dashed]
  table[row sep=crcr]{%
1	0.099352576\\
2	0.0202816129607926\\
3	0.00137189104847751\\
4	0.000336178709150183\\
5	8.38791020090510e-05\\
6	2.11980555561062e-05\\
7	1.59279240845934e-06\\
8	4.15517798057321e-07\\
9	1.08683810779229e-07\\
10	8.27510707965875e-09\\
};
\addlegendentry{$f_{{\rm qom}}(\mathbf{q}^*)$}

\end{axis}
\end{tikzpicture}%
\caption{Finite-time performance of $f^{*}_{z_s}, f_{\rm trim}$,  and $f_{{\rm qom}}(\bf q^*)$  in the absence of the adversary. }
\label{Fig:finiteCompare}
\end{figure}

}

\section{Conclusion and Future Work}
\label{sec:conclusion}
In this paper, we studied the trade-off between the detection performance of a detector when the attacker is absent (termed efficiency) and the ``worst-case'' detection performance when the attacker, knowing the detector, is present (termed security).
The setting is that a  binary hypothesis testing
is conducted based on measurements from a set of sensors, some of which can be compromised by an attacker and their measurements can be manipulated arbitrarily.
We first provided the fundamental limits of the trade-off between the efficiency and security of any detector.
We then presented detectors that possesses the limits of the efficiency and security.
Therefore, a clear guideline on how to balance the efficiency and security has been established for the system operator.
An interesting point of the fundamental trade-off is that in some cases, the maximal efficiency and the maximal security can be achieved simultaneously, i.e., the maximal efficiency (security) can be achieved without compromising any security (efficiency).
In addition, two extensions were investigated: secure sensors are assumed for the first one, and the detection performance beyond the efficiency and security is concerned for the second one.
The main results were verified by numerical examples.
Investigating the problem when the measurements from the benign sensors are not i.i.d. is a future direction.

\appendices
\section{The Proof of Theorem~\ref{theorem:ProperRateFunc}}
\label{Proof:ProperRateFunc}
The following lemma is needed to prove Theorem~\ref{theorem:ProperRateFunc}:
\begin{lemma}
If Assumption~\ref{assumpt:0interior} and \ref{assumpt:finiteKLs} hold, then the following statement is true:
\begin{enumerate}
  \item For any $w$,
\begin{align}
  \label{eqn:M0M1}
M_0(w+1) =M_1(w).
\end{align}
\item There exists a small enough $\epsilon > 0$, such that $\log M_0(w)$ is well-defined on $[-\epsilon, 1+\epsilon]$, and $\log M_1(w)$ is well-defined on $[-1-\epsilon, \epsilon]$.
\item $\log M_0(w),\,\log M_1(w)$ are strictly convex.
\item The derivative of $\log M_0(w)$ and $\log M_1(w)$ satisfy
\begin{align}
  (\log M_0(w))^{(1)}|_{w=1} &= D(1\|0). \label{eqn:firstDeriva}\\
  (\log M_0(w))^{(1)}|_{w=0} &= -D(0\|1).\label{eqn:firstDerivb}\\
  (\log M_1(w))^{(1)}|_{w=0} &= D(1\|0).\label{eqn:firstDerivc}\\
  (\log M_1(w))^{(1)}|_{w=-1} &= -D(0\|1).\label{eqn:firstDerivd}
\end{align}
\end{enumerate}
\end{lemma}
\begin{proof}
  By definition,
\begin{align*}
  M_0(w+1) &= \int_{-\infty}^\infty \left( \frac{\rm d\mu}{\rm d\nu}(y)\right)^w \frac{\rm d\mu}{\rm d\nu}(y){\rm d}\nu(y)\\
&= \int_{-\infty}^\infty \left(\frac{\rm d\mu}{\rm d\nu}(y) \right)^w {\rm d\mu}(y) = M_1(w) ,
\end{align*}
which proves \eqref{eqn:M0M1}.

 Assuming $M_0(w_1), \, M_0(w_2) < \infty$ and $w_1 < w_2$, by the convexity of the exponential function, we know that for any $\lambda$ and $0<\alpha,\beta < 1$ and $\alpha + \beta = 1$,
  \begin{align*}
    0 < \exp\left[(\alpha w_1+\beta w_2)\lambda\right]\leq \alpha e^{w_1\lambda}+\beta e^{w_2\lambda}.
  \end{align*}
  Therefore $0< M_0(\alpha w_1 + \beta w_2) \leq \alpha M_0(w_1)+\beta M_0(w_2)$ is well-defined, which proves that the domain of $\log M_0(w)$ is convex.

Furthermore, by Assumption~\ref{assumpt:0interior}, $0\in\inte(\dom_1)$ gives
\begin{align} \label{eqn:1interi0}
  1\in \inte(\dom_0),
\end{align}

Hence, $[0,1] \subset \inte(\dom_0)$, which proves that $\log M_0(w)$ is well-defined on $[-\epsilon,1+\epsilon]$ if $\epsilon$ is small enough.

It is well known that $\log M_0(w)$ is \emph{infinitely differentiable} on $\inte(\dom_0)$ (see~\cite[Exercise 2.2.24]{dembo2009large}).
Basic calculations give that
\begin{align}
  &\quad(\log M_0(w))^{(2)}  \\
  &= \int_{\R}\left(\frac{\rm d\mu}{\rm d \nu}(y) \right)^w \left(\log\left(\frac{\rm d\mu}{\rm d \nu}(y) \right)\right)^2 \rm d\nu(y) >0 \label{eqn:secondDerivative}
\end{align}
always holds, where $(\log M_0(w))^{(2)}$ is the second derivative. The above quantity is strictly positive since the KL divergence between probability measure $\mu$ and $\nu$ are strictly positive by Assumption~\ref{assumpt:finiteKLs}. Therefore, $\log M_0(w)$ is strictly convex.

The domain and the strict convexity of $\log M_1(w)$ can be proved similarly.

Take the derivative of $\log M_0(w)$ at $w = 1$ yields
\begin{align*}
  (\log M_0(w))^{(1)}|_{w=1} &= \int_{\R}\lambda(y)\frac{\rm d\mu}{\rm d \nu}(y) {\rm d\nu}(y)  \\
&=\int_{\R}\lambda(y)  { \rm d\mu}(y) = D(1\|0).
\end{align*}
Equation \eqref{eqn:firstDerivb}, \eqref{eqn:firstDerivc} and \eqref{eqn:firstDerivd} can be proved similarly.
\end{proof}
We are now ready to prove Theorem~\ref{theorem:ProperRateFunc}:
\begin{proof}[Proof of Theorem~\ref{theorem:ProperRateFunc}]
Define the derivative of $\log M_0(w)$ to be $\psi(w)$. Since $\log M_0(w)$ is strictly convex, we know that $\psi(w)$ is strictly increasing and therefore, its inverse function is well defined on $[-D(0\|1), D(1\|0)]$.
Denote the inverse function as $\varphi(x)$.
By the convexity of $\log M_0(w)$, we have that
\begin{align}
  \log M_0(w) \geq  \log M_0(w_*) + \psi (w_*) (w-w_*).
  \label{eqn:strictconvex}
\end{align}
Hence, for any $x \in [-D(0\|1), D(1\|0)]$, suppose that $\psi(w_*) = x$, we have
\begin{align*}
 wx &- \log M_0(w) = \left[w_* \psi(w_*) -\log M_0(w_*)\right]\\
& + \left[(w-w_*) \psi(w_*) + \log M_0(w_*) - \log M_0(w)\right].
\end{align*}
Notice the last term on the RHS of the equation is non-positive. Hence, we can prove that
\begin{align}
 I_0(x) =  w_* \psi(w_*) -\log M_0(w_*) = \varphi(x)x -\log M_0(\varphi(x)).
  \label{eqn:I0explicit}
\end{align}

Take the derivative and second order derivative of $I_0(x)$ we have
\begin{align*}
 \frac{{\rm d}I_0(x)}{{\rm d}x} = \varphi(x),\, \frac{{\rm d}^2I_0(x)}{{\rm d}x^2} = \frac{1}{\psi^{(1)}(\varphi(x))} > 0,
\end{align*}
where the last inequality is due to the fact that $\log M_0(w)$ is strictly convex, and thus its second derivative $\psi^{(1)}$ is strictly positive. Hence we prove that $I_0(x)$ is twice differentiable and strictly convex on $[-D(0\|1), D(1\|0)]$. Notice that
\begin{align*}
 \left.\frac{{\rm d}I_0(x)}{{\rm d}x}\right|_{x = -D(0\|1)} = \varphi(-D(0\|1)) = 0,
\end{align*}
we can prove that $I_0(x)$ is also strictly increasing. Similarly we can prove the properties for $I_1(x)$.

Combining \eqref{eqn:I0explicit}, \eqref{eqn:firstDeriva} and \eqref{eqn:firstDerivb}, we can prove \eqref{eqn:RateFuncKL00} and \eqref{eqn:RateFuncKL0}. Equation \eqref{eqn:RateFuncKL10} and \eqref{eqn:RateFuncKL1} can be proved similarly.

Since
\begin{align*}
  I_0(0) &= \sup_w 0\cdot w -\log M_0(w) = \sup_w -\log M_0(w),\\
  I_1(0) &= \sup_w 0\cdot w -\log M_1(w) = \sup_w -\log M_1(w),
\end{align*}
and
\begin{align*}
   M_0(w + 1) =  M_1(w),
\end{align*}
We can conclude $I_0(0) = I_1(0)$.
\end{proof}
\section{The Proof of Theorem~\ref{theorem:PerformanceLimitation}}
\label{Proof:FundamentalLimit}
The proof is divided into four parts, each of which is devoted to one of the statements in Theorem~\ref{theorem:PerformanceLimitation}.

{\bf Part \Rmnum{1}.}
For any index set $\mathcal{O} \subset \M$ and $\chi\in\R$, define the following Bayesian like detector:
\begin{align}  \label{eqn:fMAP1}
  f_{k,\chi,\mathcal{O}}(\mathbf Y'(k)) = \begin{cases}
    0 & \text{if } \sum_{i\in\mathcal{O}} \bar\lambda_i(k) <\chi,\\
    1 & \text{if } \sum_{i\in\mathcal{O}} \bar\lambda_i(k) \geq \chi,
  \end{cases}
\end{align}
where $\bar\lambda_i(k)$ is the empirical mean of the log-likelihood ratio from time $1$ to $k$ for sensor $i$, {which is defined in~\eqref{eqn:barlambda}.}
Denote
\[ f_{\chi,\mathcal{O}} = (f_{1,\chi,\mathcal{O}}(\mathbf Y'(1)),\cdots,f_{k,\chi,\mathcal{O}}(\mathbf Y'(k)),\cdots)  \]
and
\begin{align} \label{eqn:fStar}
  f^*_{\mathcal{O}} \triangleq f_{0,\mathcal{O}}.
\end{align}
It is well known that $f^*_{\M}$ minimize the average error probability~\cite{key1993fundamentals}: $e(\theta=0,\mathcal O=\emptyset,k)+e(\theta=1,\mathcal O=\emptyset,k)$, where, recall, $e(\theta,\mathcal{O},k)$ is defined in~\eqref{eqn:errorProb}.
Notice that
\begin{align*}
  &\liminf_{k\rightarrow\infty} -\frac{\log (e(\theta=0,\mathcal O=\emptyset,k)+e(\theta=1,\mathcal O=\emptyset,k))}{k}\\
  = &\liminf_{k\rightarrow\infty} -\frac{\log \max_{\theta} e(\theta,\mathcal O=\emptyset,k)}{k}.
\end{align*}
Hence, when the attacker is absent, $f^*_{\M}$ is optimal  in the sense that the rate $\rho$ defined in~\eqref{eqn:RateObj} is maximized.
Furthermore, Cram\'{e}r's Theorem gives that
$\E(f^*_{\M})  = mI_0(0) = mC$.
Therefore,
$\E(f)\leq mC$ holds for any detector $f$.

{\bf Part \Rmnum{2}.} In this part, we show $\bS(f)\leq(m-2n)^{+}C$.
The proof is by construction: we construct a attack strategy $g^*$ such that, for any detection rule $f$, the following inequality holds:
\begin{align} \label{eqn:universalAttack}
  \rho(f,g^*) \leq (m-2n)^+C.
\end{align}
Let $\mathcal{O}'=\{1,\ldots,n\}$ and $\mathcal{O}'' =\{m-n+1,\ldots,m\}$.
The attack strategy $g^*$ is as follows.
\begin{itemize}
\item[$(\rmnum{1})$.] When $\theta=0$, sensors in $\mathcal{O}'$ are compromised and the distributions are flipped, i.e., the measurements of sensors in $\mathcal{O}'$ are i.i.d. as $\mu$.
\item[$(\rmnum{2})$.]  When $\theta=1$, sensors in $\mathcal{O}''\setminus\mathcal{O}'$ are compromised and the distributions are flipped.
\end{itemize}
Thus, under attack $g^*$, for either $\theta=1$ or $\theta=0$, sensors in $\mathcal{O}'$ will follow distribution $\mu$ and sensors in $\mathcal{O}''\setminus\mathcal{O}'$ will follow distribution $\nu$.
In other words, only sensors in $\M\setminus(\mathcal{O}'\cup\mathcal{O}'')$ have different distributions under different $\theta$.
Notice that when $m\leq 2n$, $\M\setminus(\mathcal{O}'\cup\mathcal{O}'')=\emptyset$, which means that $\rho(f,g^*)=0$.
If $m> 2n$, by the optimality of the detection rule $f^*_{\M\setminus(\mathcal{O}'\cup\mathcal{O}'')}$ defined in~\eqref{eqn:fStar}, one obtains $\rho(f,g^*) \leq (m-2n)C$.
Equation~\eqref{eqn:universalAttack} is thus obtained.

{\bf Part \Rmnum{3}.} It is clear from the definitions of $\E(f)$ and $\bS(f)$ that $\bS(f)\leq \E(f)$ holds.

{\bf Part \Rmnum{4}.} Consider the following product measures:
\begin{align*}
  \mu_{a} &= \underbrace{\mu\times\mu\ldots\times \mu }_{m-n}\times \underbrace{\nu \times \nu \ldots \times \nu}_{n}, \\
  \mu_a^k &= \underbrace{ \mu_a\times \mu_a\ldots\times\mu_a}_{k}
\end{align*}
and
\begin{align*}
  \nu_{*}& =\underbrace{ \nu\times \nu\ldots\times\nu}_{m}, \\
  \nu_{*}^k&=\underbrace{ \nu_*\times \nu_*\ldots\times\nu_*}_{k}.
\end{align*}
The measure $\mu_a$ is generated by an attack that flips the distribution on the last $n$ sensors, when the true hypothesis is $\theta = 1$.
The measure $\nu_{*}$ is generated by benign sensors when the true hypothesis is $\theta = 0$.

Now let us consider the following problem: given $\phi >0$, find the detection rule $f$ such that
\begin{align}
  \mathbb E_{\nu_{*}^k} f_k  + \phi^k\mathbb E_{\mu_{a}^k} (1-f_k)
  \label{eq:npcost}
\end{align}
is minimized for every $k\geq 1$.
Let $f_{\phi}=(f_{\phi,1},\ldots,f_{\phi,k},\ldots)$ with $f_{k,\phi}$ given by
\begin{align} \label{eqn:fMAP}
  f_{k,\phi}(\mathbf Y'(k)) = f_{k,-\log \phi,\{1,2,\ldots,m-n\}}(\mathbf Y'(k)),
\end{align}
where, recall, the function $f_{k,\chi,\mathcal{O}}(\mathbf Y'(k))$ is defined in~\eqref{eqn:fMAP1}.
Then by the Bayesian decision-theoretic detection theory, $f_{\phi}$ is a solution to the above problem.
Let
\[ E_\phi\triangleq\liminf_{k\to\infty} - \frac{\log  \mathbb E_{\nu_{*}^k} f_{\phi,k}}{k}
\]
and
\[ S_\phi\triangleq\liminf_{k\to\infty} - \frac{\log  \mathbb E_{\mu_{a}^k} (1-f_{\phi,k})}{k}.
\]
Then from the optimality of $f_{\phi}$, for any $\phi>0$ and any detector $f=(f_1,\ldots,f_k,\ldots)$, the following hold for any $k$:
\begin{align*}
  \text{If }  \mathbb E_{\nu_{*}^k} f_k  \leq  \mathbb E_{\nu_{*}^k} f_{\phi,k},
  \text{ then }  \mathbb E_{\mu_{a}^k} (1-f_k)  \geq \mathbb E_{\mu_{a}^k} (1-f_{\phi,k}).
\end{align*}
This implies that
\begin{align*}
  &\text{If } \liminf_{k\to\infty} - \frac{\log  \mathbb E_{\nu_{*}^k} f_k}{k} \geq E_\phi,\\
  &\qquad\text{then } \liminf_{k\to\infty} - \frac{\log  \mathbb E_{\mu_{a}^k} (1-f_k)}{k} \leq S_\phi.
\end{align*}
Furthermore, the definitions of $\E(f)$ and $\bS(f)$ yield
\begin{align*}
  \E(f) \leq& \liminf_{k\to\infty} - \frac{\log  \mathbb E_{\nu_{*}^k} f_k}{k},\\
  \bS(f) \leq& \liminf_{k\to\infty} - \frac{\log  \mathbb E_{\mu_{a}^k} (1-f_k)}{k}.
\end{align*}
Therefore,
for any $\phi>0$ and any detector $f$, the following hold:
\begin{align*}
  \text{If } \E(f) \geq E_\phi, \text{then } \bS(f) \leq S_\phi.
\end{align*}

Now let us evaluate $E_\phi$ and $S_\phi$.
Let $\tilde{\phi} = -\log\phi/(m-n)$, then Cram\'{e}r's theorem yields
\begin{align*}
  E_\phi = \begin{cases}
    0 & \text{if}\, \tilde{\phi}\leq -D(0\|1),\\
    (m-n)I_0(\tilde{\phi}) & \text{if}\, \tilde{\phi} > -D(0\|1),
  \end{cases}
\end{align*}
and
\begin{align*}
  S_\phi = \begin{cases}
    0 & \text{if}\, \tilde{\phi} \geq D(1\|0),\\
    (m-n)I_1(\tilde{\phi}) & \text{if}\, \tilde{\phi} < D(1\|0).
  \end{cases}
\end{align*}
Notice that the monotonicity of  $I_0(\cdot)$ on $[-D(0\|1),\infty)$ implies that if $0<E_\phi< (m-n)I_0(D(1\|0)) = (m-n)D(1\|0)$, $\tilde{\phi}\in(-D(0\|1),D(1\|0))$ holds.
Therefore, if $0<E_\phi< (m-n)D(1\|0)$, there holds
\[S_\phi = (m-n)I_1(I_0^{-1}(E_\phi/(m-n))).\]
One thus obtains that for any detector $f$, if $0<\E(f)< (m-n)D(1\|0)$
\begin{align}
  \bS(f) \leq (m-n)I_1(I_0^{-1}(\E(f)/(m-n))).
  \label{eqn:ProofE-S1}
\end{align}
Also, it is easy to see that if $E_\phi\geq (m-n)D(1\|0)$, $S_\phi = 0$ holds.
Thus,
\begin{align} \label{eqn:ProofE-S0}
  \bS(f) = 0 \text{ if } \E(f)\geq (m-n)D(1\|0).
\end{align}

Similarly, one considers the detection problem for the measures $\mu_{*}^k$ and $\nu_{a}^k$ and obtains that
for any detector $f$, if $0<\E(f)< (m-n)D(0\|1)$
\begin{align} \label{eqn:ProofE-S2}
  \bS(f) \leq (m-n)I_0(I_1^{-1}(\E(f)/(m-n)))
\end{align}
and
\begin{align} \label{eqn:ProofE-S3}
  \bS(f) = 0 \text{ if } \E(f)\geq (m-n)D(0\|1).
\end{align}
Then equation~\eqref{E-S01} follows from~\eqref{eqn:ProofE-S1} and~\eqref{eqn:ProofE-S2}, and equation~\eqref{E-S02} from~\eqref{eqn:ProofE-S0} and~\eqref{eqn:ProofE-S3}.

\section{The Proof of Theorem~\ref{theorem:achieve}} \label{sec:proofAchieve}
{This theorem is proved by showing that $f^*_{z_s}=0$ (or $1$) if certain conditions are satisfied (i.e., Lemma~\ref{lemma:finalsubset}). Furthermore, the special structure of these conditions can ensure that, under \emph{any} attacks, $f^*_{z_s}=0$ (or $1$) if the measurements of sensors in an attack free environment belong to  a certain set, to which the Cram\'{e}r's Theorem is applied.}

Before proceeding, we need to define the following subsets of $\mathbb R^m$:
\begin{definition}
  Define $\mathcal B^-,\,\mathcal B^+\subset \mathbb R^m$ as
  \begin{align*}
    \mathcal B^{-} \triangleq \left\{\lambda\in \mathbb R^m:\sum_{i=1}^m \lambda_i < 0\right\}, \mathcal B^{+} \triangleq \left\{\lambda\in \mathbb R^m:\sum_{i=1}^m \lambda_i \geq 0\right\}.
  \end{align*}
\end{definition}

\begin{definition} \label{def:ball}
  Let $\mathcal{O}\subset \M$, $j\in\{0,1\}$ and $z\in\R_{+}$, define a ball as
  \begin{align*}
    \bal(\mathcal{O},j,z) = \Big\{\lambda\in\mathbb{R}^{m}: \sum_{i\in\mathcal{O}}I_j(\lambda_i)< z\Big\}.
  \end{align*}
\end{definition}

\begin{definition} \label{def:eball}
  Let $j\in\{0,1\}$ and $z\in\R_{+}$, define an extended ball as
\begin{align*}
  \ebal(j,z,n) \triangleq \bigcup_{|\mathcal O| = m-n}\bal(\mathcal O,j,z) .
\end{align*}
\end{definition}

From the definition of extended balls, it is clear that
\begin{align*}
 \begin{bmatrix}
\lambda_1&\ldots&\lambda_m
 \end{bmatrix} \in \ebal(j,z,n)
\end{align*}
if and only if the following inequality holds:
\begin{align*}
    \min_{|\mathcal{O}|=m-n}\sum_{i\in\mathcal{O}} I_j(\lambda_i) < z.
\end{align*}

Combining with the definition of $f_{z_s}^*$, we know that at time $k$, the output of $f_{z_s}^*$ is $0$ if and only if
\begin{align*}
  \bar \lambda(k) \triangleq
  \begin{bmatrix}
\bar \lambda_1(k)&\cdots&\bar \lambda_m(k)
  \end{bmatrix} \in \Lambda^-(z_s),
\end{align*}
where $\Lambda^-(z_s)$ is defined as
\begin{align*}
  \Lambda^-(z_s) \triangleq \ebal(0,z_s,n)\bigcup\left(\mathcal B^-\backslash \ebal(1,z_s,n)\right).
\end{align*}
The output is $1$ if $\bar\lambda(k)\in \Lambda^+(z_s)$, where
\begin{align*}
\Lambda^+(z_s)& \triangleq \mathbb R^m\backslash \Lambda^-(z_s) \\
&= \left(\mathcal B^+\bigcup \ebal(1,z_s,n)\right)\backslash \ebal(0,z_s,n)
\end{align*}

We first need the following supporting lemma.

\begin{lemma} \label{lemma:miniDimensionReduce}
  Given $\mathcal{O}_1,\,\mathcal{O}_2\subset\M\triangleq\{1,2,\ldots,m\}$ with $|\mathcal{O}_1\bigcap \mathcal{O}_2|=p>0$, $z\leq p D(1\|0)$, the optimal value of the following optimization problem is given by $pI_1(I^{-1}_0(z/p))$:
  \begin{equation} \label{eqn:ProbMini}
    \begin{aligned}
       \inf_{x\in\R^m}&\quad \sum_{i\in\mathcal O_1} I_1(x_i)\\
      \text{\rm s.t.}&\quad\sum_{i\in \mathcal O_2} I_0(x_i) < z.
    \end{aligned}
  \end{equation}
\end{lemma}
\begin{proof}
Since $I_1(\cdot)$ is nonnegative, $I_1(D(1\|0))=0$ and $x_i$ can take any value when $i\notin \mathcal O_2$, one can equivalently rewrite~\eqref{eqn:ProbMini} as
\begin{align*}
\inf_{x\in\R^m}&\quad \sum_{i\in\mathcal O_1\cap \mathcal{O}_2} I_1(x_i)\\
          \text{s.t.}&\quad\sum_{i\in \mathcal O_2} I_0(x_i)< z, \\
          &\quad x_i = D(1\|0), i\in \mathcal{O}_1\setminus\mathcal{O}_2.
\end{align*}
By the nonnegativity of $I_0(\cdot)$, the above equation is equivalent to
\begin{equation} \label{eqn:ProbMini1}
    \begin{aligned}
       \inf_{x\in\R^m, 0 < z'\leq z}& \quad\sum_{i\in\mathcal O_1\cap \mathcal{O}_2} I_1(x_i)\\
          \text{s.t.}& \quad \sum_{i\in \mathcal O_1\cap \mathcal{O}_2} I_0(x_i)< z', \\
          & \quad \sum_{i\in \mathcal O_2\setminus \mathcal{O}_1} I_0(x_i)\leq z-z', \\
          &\quad x_i = D(1\|0), i\in \mathcal{O}_1\setminus\mathcal{O}_2.
\end{aligned}
\end{equation}
To obtain the solution to the above equation, let us fist focus on the following optimization problem:
\begin{equation} \label{eqn:ProbMini2}
    \begin{aligned}
       \min_{x\in\R^m}&\quad \sum_{i\in\mathcal O} I_1(x_i)\\
      \text{\rm s.t.}&\quad\sum_{i\in \mathcal O} I_0(x_i) = z',
    \end{aligned}
  \end{equation}
where $\mathcal{O} = \mathcal{O}_1 \cap \mathcal{O}_2$. Denotes its optimal value by $\psi(z')$. In the following, we show that
\begin{align} \label{eqn:psi}
\psi(z') = pI_1(I^{-1}_0(z'/p)).
\end{align}
We claim that a solution to~\eqref{eqn:ProbMini2} is
  \begin{subnumcases}{x_i=}
    I^{-1}_0(z'/p) & if $i\in \mathcal O$,   \label{eqn:solution01}
    \\
    \text{whatever}        &if $i\notin \mathcal O$. \label{eqn:solution02}
  \end{subnumcases}
With this claim, ~\eqref{eqn:psi} clearly holds. In the following we show that this claim is correct.
  Equation~\eqref{eqn:solution02} is trivial.
  We then focus on~\eqref{eqn:solution01}.
  Due to the convexity of the functions $I_0(\cdot)$ and $I_1(\cdot)$, one obtains that for any $x\in\mathbb{R}^m$,
  \begin{align*}
    pI_0\left(\sum_{i\in\mathcal{O}}x_i/p\right) &\leq  \sum_{i\in\mathcal{O}} I_0(x_i), \\
    pI_1\left(\sum_{i\in\mathcal{O}}x_i/p\right) &\leq  \sum_{i\in\mathcal{O}} I_1(x_i)
  \end{align*}
  Therefore, without any performance loss, one may restrict the solution to the set $\mathbb{X}^*$ as follows:
  \[\mathbb{X}^* \triangleq \{x\in\R^m: x_{1}=x_{2}=\cdots=x_{p}\}.\]
  Then it is clear from the monotonicity of $I_0$ and $I_1$ that~\eqref{eqn:solution01} holds. This thus proves~\eqref{eqn:psi}.

Notice that $\psi(z')$ in~\eqref{eqn:psi} is decreasing with respect to $z'$.
Then the fact that $I_0(-D(0\|1))=0$ yields that~\eqref{eqn:ProbMini1} is equivalent to
\begin{align*}
  \min_{x\in\R^m}& \quad  \sum_{i\in\mathcal O_1\cap \mathcal{O}_2} I_1(x_i)\\
          \text{s.t.}&\quad \sum_{i\in \mathcal O_1\cap \mathcal{O}_2} I_0(x_i)= z, \\
          &\quad x_i = -D(0\|1), i\in \mathcal{O}_2\setminus\mathcal{O}_1, \\
          &\quad x_i = D(1\|0), i\in \mathcal{O}_1\setminus\mathcal{O}_2,
\end{align*}
which concludes Lemma~\ref{lemma:miniDimensionReduce} by~\eqref{eqn:psi}.
\end{proof}

\begin{lemma}  \label{lemma:intersect}
  Assume that $(z_e, z_s)$ are an admissible pair, then the following statements are true:
    \begin{enumerate}
    \item $\bal(\M,0,z_e)\subseteq \mathcal B^-$.
    \item $\bal(\M,1,z_e)\subseteq \mathcal B^+$.
    \item $\ebal(0,z_s,n)\bigcap \ebal(1,z_s,n) = \emptyset$.
    \item $\ebal(1,z_s,n)\bigcap \bal(\M,0,z_e) = \emptyset$.
    \item $\ebal(0,z_s,n)\bigcap \bal(\M,1,z_e) = \emptyset$.
    \end{enumerate}
\end{lemma}
\begin{proof}
1): It suffices to prove that given any $x\in\R^m$, if $x\in\mathcal B^+$, then $x\not\in\bal(\M,0,z_e)$. By the convexity of $I_0(x)$, one obtains that
\begin{align*}
\sum_{i\in\M} I_0(x_i) \geq mI_0(1/m\sum_{i\in\M}x_i)  \geq mI_0(0) = mC,
\end{align*}
where the second inequality follows from $x\in\mathcal B^+$ and the fact that $I_0(x)$ is increasing when $x\geq 0$.
Notice that by its definition, $z_e\leq mC$ holds. The proof is done.

2): This can be proved similarly to 1).

3): By the definition of $\ebal$, we need to prove that for any $\mathcal{O}_1,\mathcal{O}_2$ with $|\mathcal{O}_1|=|\mathcal{O}_2|= m-n$, $\bal(\mathcal O_1,0,z_s) \cap \bal(\mathcal O_2,1,z_s) = \emptyset$ holds. Notice that when $z\leq p D(1\|0)$, $pI_1(I^{-1}_0(z/p))$ is increasing with respect to $p$. Thus by Lemma~\ref{lemma:miniDimensionReduce}, it suffices to prove that $(m-2n)I_1(I^{-1}_0(z_s/(m-2n))) \geq z_s$, which is true because $0\leq z_s \leq (m-2n)C$, $pI_1(I^{-1}_0(z/p))$ is decreasing with respect to $z$ when $z\leq p D(1\|0)$, and $(m-2n)I_1(I^{-1}_0(0)) = (m-2n)C$.

4): Similar to 3), it suffices to prove that  for any $\mathcal{O}_1$ with $|\mathcal{O}_1|= m-n$, $\bal(\mathcal O_1,0,z_s) \cap \bal(\M,0,z_e) = \emptyset$ holds.
By Lemma~\ref{lemma:miniDimensionReduce}, it suffices to prove that $(m-n)I_1(I^{-1}_0(z_e/(m-n))) \geq z_s$. Then it is equivalent to prove that $(m-n)I_1(I^{-1}_0(h_e(z_s)/(m-n))) \geq z_s$, which follows from the definition of $h_e(z)$ and the fact that $pI_1(I^{-1}_0(z/p))$ is decreasing with respect to $z$ when $z\leq p D(1\|0)$.

5): This can be proved similarly to 4).
\end{proof}

From Lemma~\ref{lemma:intersect}, one obtains straightforwardly the following lemma.
\begin{lemma} \label{lemma:finalsubset}
  Assume that $(z_e, z_s)$ are an admissible pair, then the following set inclusions are true:
    \begin{enumerate}
    \item $\ebal(0,z_s,n)\subseteq \Lambda^-(z_s)$.
    \item $\ebal(1,z_s,n)\subseteq \Lambda^+(z_s)$.
    \item $\bal(\M,0,z_e)\subseteq \Lambda^-(z_s)$.
    \item $\bal(\M,1,z_e)\subseteq \Lambda^+(z_s)$.
    \end{enumerate}
\end{lemma}

We are now ready to prove Theorem~\ref{theorem:achieve}.
\begin{proof}[Proof of Theorem~\ref{theorem:achieve}]
We focus on the proof of $\bS(f_{z_s}^*) \geq z_s$, and a similar (and simpler) approach can be used to prove $\E(f_{z_s}^*) \geq z_e$.
Notice that $\ebal(0,z_s,n)\subseteq \Lambda^-(z_s)$ { in Lemma~\ref{lemma:finalsubset}} gives that, under any attacks, {there holds $\bal(\M,0,z_s) \subseteq \Lambda^-(z_s)$. Therefore, }
\begin{align*}
  &\quad\limsup_{k\to\infty} \frac{1}{k} \log \mathbb{P}_0(f_{z_s,k}^*=1) \\
  &\leq \limsup_{k\to\infty} \frac{1}{k} \log \mathbb{P}_0^{o}(\bar \lambda(k)\in\R^m\setminus\bal(\M,0,z_s))\\
  &\leq - \inf_{x\in\R^m\setminus\bal(\M,0,z_s)}\sum_{i=1}^mI_0(x_i)\\
  &= -z_s, \addtag \label{eqn:securityLower}
\end{align*}
where the second inequality holds because of the Cram\'{e}r's Theorem and the fact that $\R^m\setminus\bal(\M,0,z_s)$ is closed.

Similarly, by $\ebal(1,z_s,n)\subseteq \Lambda^+(z_s)$ { in Lemma~\ref{lemma:finalsubset}}, one obtains
\begin{align}\label{eqn:securityLower1}
  \limsup_{k\to\infty} \frac{1}{k} \log \mathbb{P}_1(f_{z_s,k}^*=0)
  \leq -z_s.
\end{align}
It follows from~\eqref{eqn:securityLower} and~\eqref{eqn:securityLower1} that $\bS(f_{z_s}^*) \geq z_s$.
The proof is thus complete.
\end{proof}

\section{The Proof of Theorem~\ref{theorem:SEsimultaneously}}
\label{Proof:SEsimultaneously}
Define the following two functions $h_0(z), h_1(z): \left(0,D_{\min}\right) \mapsto \left(0,D_{\min}\right)$:
\begin{align*}
  h_0(z) &= I_0(I_1^{-1}(z)), \\
  h_1(z) &= I_1(I_0^{-1}(z)).
\end{align*}
Then we have the following two lemmas on $h_0(z)$ and  $h_1(z)$.
\begin{lemma}
  Both $h_0(z)$ and  $h_1(z)$ are convex. Furthermore, the following equality holds:
  \begin{align} \label{eqn:I1I0IC}
    h_0(C) = h_1(C) = C.
  \end{align}
\end{lemma}

\begin{proof}
The equation~\eqref{eqn:I1I0IC} follows directly from~\eqref{eqn:I00equalI10}~and~\eqref{eqn:DefC}. To prove the convexity of $ h_0(z)$ and $ h_1(z)$, we first need to prove that $I_0^{-1}(x)$ is convex and $I_1^{-1}(x)$ is concave on $[0,\,D_{\min}]$.
Notice that if $\psi$ is the inverse function of $\phi$ and $\phi$ are twice differentiable, then by chain rule
\begin{align*}
  \psi^{(2)}(x) = -\frac{\phi^{(2)}(\psi(x))}{\left[\phi^{(1)}(\psi(x))\right]^3}.
\end{align*}
Therefore, since $I_0(x)$ ($I_1(x)$) is strictly convex and strictly decreasing (increasing) on $[-D(0\|1),\,D(1\|0)]$, $I_0^{-1}(x)$ ($I_1^{-1}(x)$) is convex (concave) on $[0,\,D_{min}]$.

The convexity of $ h_0(z)$ and $ h_1(z)$ then follows the fact that the composition of a convex and increasing (decreasing) function with a convex (concave) function is convex~\cite{boyd2004convex}.
\end{proof}

We are now ready to prove Theorem~\ref{theorem:SEsimultaneously}
\begin{proof}
By chain rule, we know that
  \begin{align*}
    \frac{{\rm d} h_0(z)}{{\rm d}z}|_{z = C} =I_0^{(1)}(x)|_{x=0}\times\frac{1}{I_1^{(1)}(x)|_{x=0}} = -1.
  \end{align*}
  Therefore, by the convexity of $ h_0(z)$, we know that
 \begin{align*}
    h_0(z) \geq h_0(C)-(z-C)\times \frac{{\rm d} h_0(z)}{{\rm d}z}|_{z = C}= 2C - z.
 \end{align*}
 Similarly, one can prove that
 \begin{align*}
    h_1(z) \geq 2C -z.
 \end{align*}
 Hence, by the definition of $h(z)$,
 \begin{align*}
  h(z) \geq 2(m-n)C - z,
 \end{align*}
which implies that $h((m-2n)C) \geq mC $ holds and $f^*_{(m-2n)C}$ achieves maximum security and efficiency simultaneously.
\end{proof}

\bibliographystyle{IEEEtran}
\bibliography{ref_xiaoqiang}

\end{document}